%% file: main.tex
\documentclass[12pt,a4paper,twoside]{article}
\input{myPackage.tex}

\title{Classification of symmetry protected states of quantum spin chains for continuous symmetry groups}

 \author[]{Bruno de Oliveira Carvalho\footnote{bruno.oliveira@kuleuven.be}}
 \author[]{Wojciech De Roeck\footnote{wojciech.deroeck@kuleuven.be}}
 \author[]{Tijl Jappens\footnote{tijl.jappens@hotmail.com}}
 \affil[]{Instituut voor Theoretische Fysica, KU Leuven}
\date{\today}

\newcommand{\caA}{\mathcal{A}}

\newcommand{\caP}{\mathcal{P}}
\newcommand{\bbZ}{\mathbb{Z}}
\newcommand{\bbC}{\mathbb{C}}
\newcommand{\bbN}{\mathbb{N}}
\newcommand{\bbR}{\mathbb{R}}
\newcommand{\caF}{\mathcal{F}}
\newcommand{\caI}{\mathcal{I}}

\newcommand{\diam}{\mathrm{diam}}

\renewcommand{\AA}{\mathcal A}

\newcommand{\stack}{\tilde \otimes}
\newcommand{\charge}{^{(c)}}
\newcommand{\chargep}{^{(c')}}
\newcommand{\spacing}{\vspace{.5cm}}

\numberwithin{equation}{section}

\begin{document}
	\maketitle
	\abstract{ Symmetry protected states (SPT's) of quantum spin systems were studied by several authors. For one-dimensional systems (spin chains), there is an essentially complete and rigorous understanding: SPT's corresponding to finite on-site symmetry groups $G$ are classified by the second cohomology group $H^2(G,U(1))$, as established in \cite{kapustin2021classification}.  We extend this result to the case of compact topological symmetry groups $G$. We also strengthen the existing results in the sense that our classification results holds within the class of spin chains with locally bounded on-site dimensions. }
 
 \section{Introduction}

This paper fits into the study of topological properties and characterization of ground state of quantum many body Hamiltonians. The systems are built on an extensive lattice like $\bbZ^d$, and every site of the lattice hosts a finite number of degrees of freedom. We restrict to spin systems; to every site is associated a quantum spin, corresponding to  a finite dimensional Hilbert space. 
In this field, one calls ground states \emph{trivial} if they can be continuously deformed into product states, i.e.\ states that are completely factorized between different spatial regions.  The appropriate notion of ``continuous deformation'' is intuitive but somehow subtle from a mathematical point of view and it will be introduced later. 
We imagine the spin systems to be equipped with some symmetry $G$ that acts separately for each lattice site, with the most relevant example being rotation in spin space. 
We study symmetry protected order, a concept introduced by Gu and Wen \cite{guwen2009} and further elucidated by \cite{Chen_2013,chenguwen2010, chen_gu_wen_2011}, \cite{schuch2011MatrixProduct}, \cite{pollman2012symmetry}. It means that we are interested in states that are trivial in the sense defined above, but the continuous deformation that trivializes them is possibly forbidden because it does not respect the $G$-symmetry.   
This leads to a more restrictive equivalence relation between $G$-symmetric states: We say such states are $G$-equivalent if they can be continuously deformed into each other without breaking the symmetry. As it stands, two states can be inequivalent because the number of degrees of freedom per site is different, even if some of these degrees of freedom are trival. 
To remedy this, one needs the notion of stacking. Two systems are stacked into one by taking on each site the tensor product of the Hilbert spaces associated to that site. 
The correct notion of equivalence is now 
stable $G$-equivalence. Two states are stably $G$-equivalent if they are $G$-equivalent after having been stacked with two trival states. In this way, the set of equivalence classes is equipped with a binary relation that can be shown to be a group. The equivalence classes are called SPT-states and they have physically interesting features that can be characterized by the behaviour at spatial boundaries. The best known example is the appearance of fractionalized spins at $0$-boundaries, which corresponds to the well-known \emph{Haldane phase} \cite{haldane_continuum_1983,haldane_nonlinear_1983}.  

As in many previous works on the subject, we will omit the Hamiltonian in our setup, and we focus on the states themselves. As such, the only input is the space dimension $d$ and the symmetry group $G$. There are undisputed conjectures for the case $d=1,2$: One believes that the set of stable $G$-equivalence classes, equipped with the stacking relation, is isomorphic to the Borel cohomology group $H_{\text{Borel}}^{d+1}(G,U(1))$. For dimensions larger than 2, there are SPT phases that lie beyond this cohomology picture \cite{burnell2014,vishwanath}, and there are proposals for generalized classifications \cite{Freed2014,Freed2016,kapustin2021classification,kitaevtalk,wen.anomalies} (see also \cite{Xiong_2018}).

The classification of symmetry protected states of spin chains based on projective representations of groups was firstly discussed in \cite{pollmann_entanglement_2010}. Since then, rigorous topological indexes classifying finite symmetry group SPT states have been defined in $d=1$ \cite{ogatatimereversal,ogata2019classification} for states satisfying the split property \cite{matsui.splitproperty}, and in $d=2$ for invertible states \cite{sopenkoindex}, and short-range entangled states \cite{ogata2021h3gmathbb}. If the symmetry group is finite and on-site, there is even a complete classification of SPT's in $d=1$ \cite{kapustin2021classification}. In \cite{quella}, the authors classify symmetric matrix product states with respect to simple compact Lie groups. Recently, $H^2_{\text{Borel}}(G,U(1))$-valued indices for 1d SPT states were obtained when the group $G$ is a compact Lie group, and computed explicitly for the $SO(n)$ groups \cite{Ragone:2024wie} and for $G=SO(3)$ \cite{tasaki_topological_2018}.

The present paper extends these results by giving a full classification of 1d SPT states, where the symmetry group $G$ is a topological group that is compact and metrizable.  Our treatment follows \cite{kapustin2021classification} but we restrict the on-site Hilbert spaces to have a uniformly bounded dimension (as a function of site). This forces us to revisit several of the auxiliary results.  
The construction of the map from SPT states to $H^2_{\text{Borel}}(G,U(1))$ proceeds essentially in the same way as in previous works. To show that this map is injective, we however have to develop a new idea, since the previous proofs relied in an essential way on the finiteness of the group $G$.

\subsection{Acknowledgements}
The authors were supported  by the FWO (Flemish Research Fund) grant G098919N, 
the FWO-FNRS EOS research project G0H1122N EOS 40007526 CHEQS, the KULeuven  Runners-up grant iBOF DOA/20/011, and the internal KULeuven grant C14/21/086.



\section{Setup and Main Result}   \label{sec: setup}

In this section, we introduce the technical framework. The first sections \ref{subsection.algebras}, \ref{subsection.lga} and \ref{sec: states}, are completely standard and covered in many review articles \cite{bratteliI,bratteliII,naaijkens2017quantum,brunoamandaI,nachtergaele.sims.ogata.2006,brunoamandaII}. The later sections deal with more specific aspects of short-range entangled systems. Finally, in section \ref{subsec: main result}, we state our main result. 
\subsection{Algebras} \label{subsection.algebras}
A spin chain $C^*$-algebra $\caA$ is defined in the standard way.
To any site $j \in \bbZ$, we associate an $n_j$-dimensional on-site Hilbert space $\mathcal H_j$, with associated matrix algebra $\caA_j$ isomorphic to $M_{n_j}(\bbC)$, the algebra of $n_j\times n_j$ matrices with complex entries. 
We assume that there is a $n_{\max}$ such that $n_j \leq n_{\max}$. 
The algebra $M_{n_j}(\bbC)$ is equipped with its natural operator norm and $*$-operation (Hermititian adjoint of a matrix) making it into a $C^*$-algebra.  The spin chain algebra $\caA$ is the inductive limit of 
algebra's $\caA_S=\otimes_{j\in S} \caA_j$, with  $S \in \caP_{\text{fin}}$, the finite subsets of $\bbZ$. 
It comes naturally equipped with local subalgebra's $\caA_X, X\subset \bbZ$. We refer to standard references \cite{bratteliII,simon2014statistical,naaijkens2017quantum} for more background and details.

We refer to such algebras as \emph{chain} algebras, and we will not repeat anymore that we require them to have uniformly bounded dimensions $n_j$. Given a pair of chain algebra's $\caA,
\caA'$, we consider the \emph{stacked} chain algebra 
$$\caA \stack \caA',$$ 
defined as above, i.e.\ as the inductive limit, but now starting from the stacked on-site algebras $\AA_j \stack \AA_j'$, which are simply defined as $\AA_j \otimes \AA'_j$ \footnote{The stacked chain algebra $\AA\stack \AA'$ defined this way coincides, as a $C^*$-algebra, with the $C^*$-tensor product $\AA \otimes \AA'$, with respect to the cross-norm (cf.\ \cite{turumaru}), which, by a result of \cite{takesaki.norm}, is the unique compatible norm on the algebraic direct product $\AA \odot \AA'$. I.e., the unique norm by which the completion of $\AA \odot \AA'$ is a $C^*$-algebra, and such that $\norm{A \otimes A' } \le \norm{A} \norm{A'}$.}, corresponding to the on-site Hilbert spaces  $\HH_j \stack \HH'_j := \HH_j \otimes \HH_j'$. The operation of stacking two spin chain algebra's in this way is crucial for us, and for this reason we use the distinghuished notation $\stack$.    \par
Finally, we say that spin chain algebras $\caA,
\caA'$ are isomorphic whenever their local dimensions agree, i.e.\ $n_j=n'_j$. In that case, we can pick unitaries $\mathscr F_j \in \mathcal U(\mathcal H_j, \mathcal H_{j}')$ and we can construct an isomorphism of $C^*$-algebra's $\mathscr F $:  For any  $S \in \caP_{\text{fin}}$, we define the isomorphisms $\mathscr F_S:  \AA_S \to \AA'_S$ by setting 
$\mathscr F_S: \otimes_{j\in S}\Ad{\mathscr F_j}$. These definitions are consistent and they extend to an isomorphism $\AA\to \AA'$. 


\subsection{Locally generated automorphisms and their generators}\label{subsection.lga}
We introduce the framework to discuss time-evolutions. 


\subsubsection{Interactions}\label{sec: TDIs}
Let $\caF$ be the class of non-increasing, strictly positive functions $f:\bbN^+\to\bbR^+$, with $\bbN^+=\{1,2,\ldots\}$, satisfying the fast decay condition $\lim_{r\to\infty}r^pf(r)=0$ for any $p>0$. 
An interaction is a collection of operators $H=(H(S))$ labelled by $S \in \mathcal P_{\text{fin}}(\ZZ)$ such that $H(S)=H(S)^* \in\caA_S$ and
$$
|| H ||_f  = \sup_{j\in\bbZ}  \sum_{S \in \caP_{\text{fin}}: S \ni j} \frac{||H(S)||}{f(1+\diam(S))}
$$
is finite for some $f \in\caF$.  Furthermore, we say an interaction $H$ is \textit{anchored} at a region $X \subset \ZZ$ (or $X$-\textit{anchored}) if
$$X \cap S = \emptyset \implies H(S) = 0.$$
There is a related norm on $X$-anchored interactions 
\begin{align*}
    \norm{H}_{X,f} := \begin{cases}
        \norm{H}_f,\ &H \text{ is }X\text{-anchored }, \\
        \infty,\ &\ \text{otherwise.}
    \end{cases}
\end{align*}



 
\subsubsection{Locally generated automorphisms}\label{ALEs}

We will consider families of interactions $H(s)$ parametrized by $s\in [0,1]$ and call them `time-dependent interactions' (TDIs) provided that they satisfy some regularity conditions to be formulated below. Since the risk of confusion is small, we will often denote them by $H$ as well. A time-dependent interaction $H$ is a TDI if there is $f\in\caF$ such that $s\mapsto H(s)$ is $||\cdot||_f$ bounded and strongly measurable\footnote{namely, it is the limit of a sequence of simple functions, pointwise almost everywhere, where the limit at any $s$ is taken in $||\cdot||_f$-norm, see~\cite{diestel1978vector} }.  For such functions, we use the supremum norms 
$$
||H||_{f}=\sup_{s\in[0,1]} ||H(s)||_{f}.
$$
The role of a TDI $H$ is to generate a family of \emph{locally generated automorphisms (LGA's)} $(\alpha_{H,s})_{s\in[0,1]}$, namely the one-parameter family of strongly continuous $^*$-automorphisms $\alpha_{H,s}$ on $\caA$ that should be viewed as a particular solutions of the Heisenberg evolution equations on $\caA$:
\begin{equation}\label{eq: heisenberg}
\alpha_{H,s}(A)= A +i\int_0^s du\, \alpha_{H,u} ( [H(u),A] ).
\end{equation}
The expression $[H(u),A]$ on the right-hand side is defined as  $\sum_{S \in \caP_{\text{fin}}}  [H(u,S),A]$ and it is finite for $A$ in a dense subset of $\AA$. 
 The integral on the right hand side is to be understood in the sense of Bochner, and the strong measurability of the integrand follows from strong measurability of $s\mapsto H(s)$ and the strong continuity of $s\mapsto \alpha_{H,s}$. The existence (let alone: uniqueness) of solutions to \eqref{eq: heisenberg} is highly non-trivial. 
Concretely, $\alpha_{H,s}(A)$ is defined as the limit (in the topology of $\caA$) of solutions of~(\ref{eq: heisenberg}) where $H$ is replaced with a finite volume restriction $H^{(\Lambda_n)}$ given by  $H^{(\Lambda_n)}(s,S)= \chi(S \subset \Lambda_n)H(s,S)$ for an increasing and absorbing sequence $(\Lambda_n)_{n\in\bbN}$ of subsets of $\mathbb{Z}$. By a standard argument using the Lieb-Robinson bound \cite{liebrobinson,nachtergaele.sims.ogata.2006}, this procedure is well-defined and the limit solves~(\ref{eq: heisenberg}) for $A$ in a dense set.  We conveniently define the notation
\begin{equation}
    \alpha_H := \alpha_{H,1}.
\end{equation}

\noindent We list a series of properties of LGA's. Their proofs are standard, see e.g.\ \cite{wojciech.thoulesspump,sopenkoindex}.
\begin{lemma} \label{lem.lgaproperties}
    Let $H,\ H_1,\ H_2$ be TDI's with finite $f$-norm. Then
    \begin{enumerate} [label=(\roman*)]
        \item $((\alpha_{H,s})^{-1})_{s \in [0,1]}$ is a family of LGA's generated by a TDI with finite $\tilde f$-norm, for $\tilde f \in \mathcal F$ depending only on $f$,

        \item ($\alpha_{H_1,s} \circ \alpha_{H_2,s})_{s \in [0,1]}$ a family of LGA's generated by a TDI with finite $\tilde f$-norm, for $\tilde f \in \mathcal F$ depending only on $f$. 
    \end{enumerate}

\end{lemma}

In the case where a TDI $H$ is anchored in a finite set $X$, i.e.\
$||H||_{X,f}<\infty$, 
the corresponding family of LGA's $\alpha_{H,s}$ has some nice and intuitive properties, making it similar to quantum dynamics in finite volume.  Most notably, 
$$
\alpha_{H,s}=\Ad{W(s)}
$$
where  the unitary $W(s) \in \AA$ solves the equation
$$
W(s) = i\int_0^s du   W(u)Z(u), \qquad W(0)=\mathds 1
$$
with 
$$
Z(u)=\sum_{S} H(s,S),
$$
being a Hermitian element of $\AA$ because of the anchoring.
Moreover, both $Z(s),W(s)$ satisfy good locality bounds, as we spell out in Appendix \ref{appendix.lga}.


\subsection{States}\label{sec: states}

States are normalized positive linear functionals on the spin chain algebra $\caA$. A natural metric on states is derived from the Banach space norm
 \begin{equation}\label{eq: metric on states}
 ||\psi-\psi'|| = \sup_{A \in\caA, ||A||=1} |\psi(A)-\psi'(A)|.
 \end{equation}
The set of states on $\AA$ is convex and its extremal points are called the pure states. A state $\psi$ on a chain algebra $\AA$ will be, from now on, denoted by $(\psi,\AA)$. A distinguished class of pure states is that of spatial \emph{product states} $(\phi,\AA)$, i.e.\ states that satisfy
$$
\phi(A_jA_i)= \phi(A_j)\phi( A_i),\qquad  A_j \in\caA_j,A_i\in \caA_i, i\neq j.
$$
A pure product state $\phi$ on $\caA$ is fully characterized by its on-site restrictions $\phi|_{\{j\}}$ to $\caA_{\{j\}}$. Each restriction $\phi|_{\{j\}}$ is represented by a density matrix in $\AA_{\{j\}}$, and purity of $\phi$ implies the range of this density matrix is contained in a one-dimensional subspace $V_j \subset \HH_j$. For each $j$, one can then choose $\xi_j \in V_j$ such that $\phi|_{\{j\}}=\langle \xi_j, (\cdot) \xi_j\rangle $. We will allow ourselves to use the formal notation
\begin{align} \label{eq.formallyproductstate}
    \phi  = \bigotimes\limits_{j\in \ZZ} \langle \xi_j, (\cdot) \xi_j\rangle.
\end{align}
States $(\psi, \AA)$ and $(\psi',\AA')$ can be stacked into a state $(\psi \stack \psi',\AA \stack \AA')$ defined by 
$$\psi \stack \psi' (A \stack A') := \psi(A) \psi'(A'), \qquad A \in \AA,\ A'\in \AA'.$$

\subsubsection{Equivalence of states}\label{sec: equivalence of states}

We say that a pair of pure states $(\psi,\AA),\  (\psi',\AA)$ are equivalent if there exists a TDI $H$, as defined in Section \ref{sec: TDIs}, such that 
$$
\psi'=\psi\circ \alpha_{H}. 
$$

\noindent This relation is clearly reflexive. It is also symmetric and transitive, by Lemma \ref{lem.lgaproperties}. 

\begin{definition} \hfill
\begin{enumerate}
    \item  A  pure state $(\psi,\AA)$ is called short-range entangled (SRE) if it is equivalent to a pure product state on the chain algebra $\caA$.
    \item  A pure state  $(\psi,\AA)$ is stably short-range entangled (stably SRE) if there exists a pure product state $(\phi,\AA')$ such that the pure state $(\psi \stack \phi, \AA \stack \AA')$ is SRE. 
\end{enumerate}
 
\end{definition}

\noindent In some cases, we need to be more specific about the locality properties of a given state: An SRE (respectively, stably-SRE) state is called \textit{$f$-SRE ($f$-stably-SRE)} for some $f \in \mathcal F$ whenever it can be produced from a product state via a TDI $H$ satisfying
$\norm{H}_f <\infty.$

\subsection{Symmetries} \label{section.gactions}

 We consider topological groups $G$ that are compact and metrizable (e.g.\ finite groups, compact Lie groups,\ldots).
Given a chain algebra $\caA$ as above, we consider an on-site unitary $G$-action $U$ which is defined as a collection $$U=\{U_j(g), j\in \bbZ, g\in G\},$$
 where $U_j(g) \in \UU(\HH_j)$ and, for each $j$,  $g\mapsto U_j(g)$ is a continuous representation of $G$. 
The unitary action $U$ gives rise to a strongly continuous family of *-automorphisms $\beta_g^X$ on $\caA_X$, $X \subseteq \ZZ$, in the following way: if $A \in \AA_S$ for $S \in \mathcal P_{\text{fin}}(X)$, then
$$\beta_g^X(A) := \bigotimes\limits_{j \in S} \Ad {U_j(g)} (A),$$
and this is extended to $\AA_X$ by linearity and density.  Here,  we used the notation $\Ad {B}(A)=BAB^*$. When $X = \ZZ$, we denote $\beta_g^\ZZ = \beta_g$.

We say that a state $(\psi,\caA)$ is $G-$invariant if $\psi=\psi\circ \beta_g$ for all $g\in G$. Obviously, this notion depends on the $G$-action $U$, which motivates the next definition and vocabulary.

 \subsubsection{G-states, G-product states and special G-product states}\label{sec:productstates}

\begin{definition} [G-states]
A $G$-state is a triple $(\psi,\caA,U)$, where 
\begin{enumerate}
    \item $\caA$ is a chain algebra,
        \item $U$ is an on-site unitary $G$-action, as defined above,
    \item $\psi$ is a pure $G$-invariant state on $\caA$ that is stably-SRE.

\end{enumerate}
\end{definition}
 A $G$-state is a $G$-product state whenever it is a product state, as defined in section \ref{sec: states}. $G$-invariance with respect to the unitary on-site group action $U$ implies each one-dimensional subspace $V_j$ characterizing the $G$-product state (cf.\ Equation \ref{eq.formallyproductstate}) is actually a one-dimensional representation of $G$, transforming according to a continuous group homomorphism $q_j \in \hom(G,U(1))$ (group homomorphisms between $G$ and $U(1)$). Such a $q_j \in \hom(G,U(1))$ will be henceforth referred to as the on-site $G$-charge. 

\begin{definition} [special G-product states]
A $G$-state $(\phi,\caA,U)$ is a special $G$-product state if it is a $G$-product state such that all on-site $G$-charges are trivial. I.e.\ 
$$q_j(g) = 1,$$
for all $j \in \ZZ$, $g\in G$.

\end{definition}

 When stacking $G$-states, we need to stack as well the $G$-actions. This is done in an obvious way: If $U,U'$ are $G$-actions on $\caA,\caA'$, we let $U\stack U'$ be the collection of on-site unitary representations  $g \mapsto U_j(g)\tilde \otimes U'_j(g)$.

\subsubsection{G-equivariant TDIs} \label{sec.trivialTDI}

A TDI $H$ is said to be $G$-equivariant if $\beta_g (H(s,S)) = H(s,S),$ for any $g\in G,\ s\in [0,1],\ S\in \mathcal P_{\text{fin}}(\ZZ)$.  If a TDI is $G$-equivariant, it generates a G-equivariant family of LGA's. Namely, for each $s \in [0,1]$, the LGA $\alpha_{H,s}$ satisfies: 
$$\alpha_{H,s} \circ \beta_g = \beta_g \circ \alpha_{H,s}, \qquad \forall g \in G.$$

Often, we will encounter rather trivial TDI's. Consider for example a family of $G$-equivariant unitaries $ (V_{I})_{I\in \caI}$ with $V_I \in \caA_I$ and the collection
 $\caI$ is a partition of $\mathbb{Z}$ into discrete intervals $I$, such that their lenghts are bounded: $\sup_{I \in \mathcal I} |I| <\infty$. 
 In this case, the formal product $\otimes_{I\in \caI}\Ad{V_I}$ is a well-defined $G$-equivariant LGA, since we can construct a $G$-equivariant Hermitian operator $H(I)$ such that $V_I=e^{-i H(I)}$ and $||H(I)|| \leq \pi$. The corresponding TDI $H(s,I)=H(I)$ (and $H(s,S)=0$ whenever $S\not\in \caI$), has finite $f$-norm, for any $f\in \caF$.



\subsection{G-Equivalence of G-states}
We refine the equivalence relation introduced in Section \ref{sec: equivalence of states} by taking symmetry into account.
\begin{definition}
Two $G$-states $(\psi,\caA,U)$ and $(\psi',\caA',U')$ are G-equivalent whenever the following conditions are satisfied:
\begin{enumerate}
    \item 
    For any site $j \in \ZZ$, the on-site  Hilbert spaces dimensions are equal, $n_j=n'_j$.  
 \item For any site $j$, there is a unitary $ \mathscr F_j: \mathcal H_j\to \mathcal H'_j$ that intertwines
 the $G$-actions: $\mathscr F_j U_j= U_{j}' \mathscr F_j$.
 \item There is a TDI $H$ on $\caA$ that is $G$-equivariant and such that the corresponding LGA connects the states
    $$
    \psi' \circ \mathscr F = \psi \circ \alpha_H
    $$
    with 
    $\mathscr F = \otimes_j\Ad{\mathscr F_j}$ as defined in subsection \ref{subsection.algebras}. 
\end{enumerate}
\end{definition}


To compare states that are not necessarily defined on isomorphic chain algebras, we introduce the notion of \emph{G-stable equivalence}.

\begin{definition} \label{def.gequivalence}
A pair of $G$-states $(\psi,\caA,U)$ and $(\psi',\caA',U')$ are $G$-stably equivalent whenever there are special $G$-product states $(\phi,\caA_\phi,U_\phi)$ and   $(\phi',\caA_{\phi'},U_{\phi'})$ such that 
$$
(\psi \tilde \otimes\phi,\caA \tilde \otimes\caA_\phi,U \tilde \otimes U_\phi)  
\qquad \text{and} \qquad
(\psi' \tilde \otimes\phi',\caA' \tilde \otimes\caA_{\phi'},U' \tilde \otimes U_{\phi'})  
$$ are $G$-equivalent.
\end{definition}

\subsection{Main result} \label{subsec: main result}


The set of stably $G$-equivalence classes of $G$-states is denoted by $\text{SPT}_G$.  It is equipped with a commutative binary relation given by stacking $\stack$. It is clear that the class containing special $G$-product states acts as an identity element for this relation.  Our result states that $\text{SPT}_G$ is a group and it is isomorphic to the abelian group $H^2_{\text{Borel}}(G,U(1))$ (Mackey-Moore cohomology group \cite{moore64,mooreII,Mackey1957BorelSI}; see also \cite{Cattaneo.mackeymoorecohomology}).

    \begin{theorem} \label{thm.grouphomomorphism}
        Given a compact and metrizable topological group G, there is a group isomorphism between $\text{SPT}_G$ and $H^2_{\text{Borel}}(G,U(1))$.
    \end{theorem}
This theorem differs from previous results mainly through its scope.
As already mentioned, the works \cite{kapustin2021classification} and \cite{ogata2019classification} established a similar result for finite groups $G$, where the cohomology group $H^2_{\text{Borel}}(G,U(1))$ resumes to the standard abelian group cohomology group (also referred to as Eilenberg-MacLane cohomology group \cite{Cattaneo.mackeymoorecohomology,eilenberg.maclane.cohomology}). The inclusion of more general groups $G$ allows to include some very natural cases, most notably $G=U(1)$ or $G=SO(3)$.   The overall strategy of our proof is very similar to the one in \cite{kapustin2021classification} and many of our intermediate results have direct analogues in  \cite{kapustin2021classification}.  
Our general setup differs from \cite{kapustin2021classification} (apart from the class of allowed groups $G$)  in that we require the on-site dimension to be bounded along the chain, which feels more natural to us.  This forces us in particular to take a different approach to the proof of Theorem \ref{thm.equivalenceofproductstates}.
However, the most substantial difference between our treatment and that of  \cite{kapustin2021classification} is in the proof of completeness of the classification in Section \ref{sec.proofmain}.  There we use a new idea that is based on properties of Schmidt vectors of the state in GNS representation. These properties are established in Section \ref{sec.technical}.

	\section{The SPT index for G-states} \label{section.index}

    In this section we define SPT indices for $G$-states. The construction is by now well-understood and we review it mainly for the sake of completeness. The fact that we allow for non-finite groups forces us anyhow to partially revisit the construction. 
 
 Since the theory of the SPT index for chains is intimately connected to projective representation, we have to recall some of that theory as well.  We refer to \cite{cattaneo} and \cite{Ragone:2024wie} for more extensive treatments. 
 
 
 \par

    \subsection{Projective representations} \label{sec.cups}
    Let $\HH$ be a separable Hilbert space, and denote by $\mathbb P(\mathcal H)$ the projective space topologized by the quotient topology. We denote by $p:\ \HH\setminus \{0\} \to \mathbb P(\HH)$ the projection map. 
    The projective space $\mathbb P(\HH)$ is a separable complete space, metrizable by the following distance function:\begin{equation}
    d(\mathbf{a}, \mathbf{b}) = \inf\limits_{ \substack{ a \in p^{-1}(\mathbf {a}) \\ b \in p^{-1}(\mathbf{b})}} \norm{\dfrac{a}{||a||} - \dfrac {b}{||b||}}, \qquad \mathbf {a},\mathbf{b} \in \mathbb P(\HH).
\end{equation}
    We also denote by $\mathscr P:\ \UU(\HH) \to \mathbb P(\UU(\HH))$ the continuous group homomorphism between the unitary group of $\HH$, equipped with strong (or, equivalently: weak) operator topology, and the projective unitary group $\mathbb P(\UU(\HH))$, topologized by the so-called final topology for $\mathscr P$: the finest topology making $\mathscr P$ continuous.
    In this setting, the group homomorphism $\mathscr P$ from $\UU(\HH)$ onto $\mathbb P(\UU(\HH))$ is continuous. Furthermore, the following essential result holds:
 \begin{proposition} [Lemma 1 of \cite{cattaneolocallycocycles}] \label{thm.borelsection}
        There exists a normalized Borel section $\Sigma$ associated to $\mathscr P$. That is, there exists a measurable right-inverse $\Sigma:\ \mathbb P(\UU(\HH)) \to \UU(\HH)$ for $\mathscr P$: 
            \begin{equation}
                \mathscr P \circ \Sigma = \text{id}_{\mathbb P (\UU(\HH))}. 
            \end{equation}
    \end{proposition}
    We are now ready to define the actual projective representations.
    \begin{definition}
    Let $G$ be a topological group.
    \hfill
    \begin{enumerate}
        \item A continuous unitary projective representation (CUP-rep) of  $G$ on $\mathbb P(\mathcal H)$ is a continuous group homomorphism
        \begin{equation}
            \mathbf U:\ G \to \mathbb P(\UU(\HH)), 
        \end{equation}
        denoted as $(\mathbf U, \HH)$. 

        \item Given a measurable map $\mu: G\times G \to U(1)$, a Borel unitary $\mu$-multiplier representation ($u_\mu,\HH$)  on $\HH$ is a measurable map
    \begin{equation}
        u_\mu:\ G \to \UU(\HH),
    \end{equation}
    such that $u_\mu (e) = \mathds 1$ and $u_\mu (g)u_\mu (h) = \mu(g,h) u_\mu(gh)$, $g,h \in G$, and $e$ is the group unit. 
        \end{enumerate}
    \end{definition}
    Moreover, a Borel unitary $\mu$-multiplier representation $(u_\mu,\HH)$ of a group $G$ is said to be a Borel lift of a CUP-rep $(\mathbf U, \HH)$ of $G$ if 
\begin{equation}
    \mathbf U (g) = \mathscr P \circ u_\mu(g), \qquad g \in G,
\end{equation}
    namely, if the following diagram commutes: 
    

\begin{center}

\tikzset{every picture/.style={line width=0.75pt}} 

\begin{tikzpicture}[x=0.75pt,y=0.75pt,yscale=-1,xscale=1]

\draw    (283.71,174) -- (340.86,174) ;
\draw [shift={(342.86,174)}, rotate = 180] [color={rgb, 255:red, 0; green, 0; blue, 0 }  ][line width=0.75]    (10.93,-3.29) .. controls (6.95,-1.4) and (3.31,-0.3) .. (0,0) .. controls (3.31,0.3) and (6.95,1.4) .. (10.93,3.29)   ;
\draw    (372.86,117.14) -- (372.86,155.14) ;
\draw [shift={(372.86,157.14)}, rotate = 270] [color={rgb, 255:red, 0; green, 0; blue, 0 }  ][line width=0.75]    (10.93,-3.29) .. controls (6.95,-1.4) and (3.31,-0.3) .. (0,0) .. controls (3.31,0.3) and (6.95,1.4) .. (10.93,3.29)   ;
\draw    (277.71,160.86) -- (344.36,114.29) ;
\draw [shift={(346,113.14)}, rotate = 145.06] [color={rgb, 255:red, 0; green, 0; blue, 0 }  ][line width=0.75]    (10.93,-3.29) .. controls (6.95,-1.4) and (3.31,-0.3) .. (0,0) .. controls (3.31,0.3) and (6.95,1.4) .. (10.93,3.29)   ;

\draw (261.57,164.26) node [anchor=north west][inner sep=0.75pt]    {$G$};
\draw (348.86,92.83) node [anchor=north west][inner sep=0.75pt]    {$\mathcal{U}(\mathcal{H})$};
\draw (347.71,165.11) node [anchor=north west][inner sep=0.75pt]    {$\mathbb{P}(\mathcal{U}(\mathcal{H}))$};
\draw (308.14,179.26) node [anchor=north west][inner sep=0.75pt]    {$\mathbf U$};
\draw (288.14,110.54) node [anchor=north west][inner sep=0.75pt]    {$u_{\mu }$};
\draw (378.14,123.26) node [anchor=north west][inner sep=0.75pt]    {$\mathscr{P}$};

\end{tikzpicture}

\end{center}
Borel unitary multipliers are intimately connected to $\text{CUP}$-reps:
\begin{lemma} [Proposition 2' of \cite{cattaneo}]
    Every Borel unitary multiplier representation of a compact topological group $G$ is the Borel lift of a $\text{CUP}$-rep and, conversely, every $\text{CUP}$-rep admits a Borel lift that is a Borel unitary multiplier representation of $G$. 
\end{lemma}
Hence we can always associate a $\text{CUP}$-rep with a Borel multiplier representation of $G$ and vice-versa. Particularly, by associativity of unitary operators on a Hilbert space, a Borel multiplier $\mu$ associated to a CUP-rep satisfies the 2-cocycle condition: 
\begin{equation*}
    \mu(g,h) \mu(gh,k) = \mu(g,hk)\mu(h,k), \qquad g,h,k \in G.
\end{equation*}
The pointwise product $\mu_1 \cdot \mu_2$ of two Borel multipliers is also a Borel multiplier. A Borel 2-coboundary is a Borel multiplier that can be written as 
$$\mu(g,h) = \nu(g) \nu(h)/\nu(gh),\qquad g,h \in G,$$
where $\nu \in \hom(G,U(1))$ is measurable. The second Borel cohomology group of $G$ with indices in $U(1)$, 
$$H^2_{\text{Borel}} (G,U(1))$$
is the abelian group of equivalence classes of Borel 2-cocycles modulo Borel 2-coboundaries, with group structure given by $[\mu_1] \cdot [\mu_2] = [\mu_1 \cdot \mu_2]$, where $[\mu_1]$ and $[\mu_2]$ are equivalence classes of 2-cocycles with representatives $\mu_1$ and $\mu_2$.
We refer to Appendix \ref{appendix.borelcohomology} for a general discussion on Borel group cohomology. 

    \subsubsection{Equivalence of projective representations} \label{sec.equivalenceofprojectivereps}

    In this section we discuss an equivalence relation for CUP-reps, which will then be classified modulo tensor products with linear representations of $G$. We say a CUP-rep $\mathbf U$ of $G$ on $\HH$ is a continuous linear representation if it has a Borel lift with trivial multiplier $\mu \equiv 1$. In this case, the lift can be chosen to be a strongly continuous unitary representation of $G$ \cite{cattaneo}. We will also use the following remark: 

    \begin{remark}
    \noindent A unitary $V \in \UU(\HH,\HH')$ between separable Hilbert spaces induces a map $\widetilde V:\ \mathbb P(\HH) \to \mathbb P(\HH')$ via
    \begin{equation} \label{eq.inducedprojectivemap}
        \widetilde V p(\xi) := p'(V\xi),\qquad \xi \in \HH,
    \end{equation}
    where $p'$ denotes the projection from $\HH'\setminus \{0\} \to \mathbb P(\HH')$. The linear map $\widetilde V$ is injective, by definition, and surjective, as $V$ is a unitary. 
    \end{remark}
    
    \begin{definition} \label{def.equivalenceofcups}
        Two continuous unitary projective representations $(\mathbf U,\HH_{\mathbf U})$ and $(\mathbf V,\HH_{\mathbf V})$ are equivalent iff.\ there exist continuous linear representations $( {\mathbf U'}, \HH_{ {\mathbf U'}})$ and $( {\mathbf V'},\HH_{ {\mathbf V'}})$ and a unitary 
           $ {W} \in \UU(\HH_{\mathbf U} \otimes \HH_{ {\mathbf U'}}, \HH_{\mathbf V} \otimes \HH_{ {\mathbf V'}})$
        such that 
        \begin{equation}
            \tilde W \circ \mathbf U \otimes  {\mathbf U'} = \mathbf V \otimes  {\mathbf V'} \circ \tilde W,
            \end{equation}
            where $\tilde W$ is the map induced by $W$, as in the previous Remark. 
    \end{definition}

    The above is an equivalence relation on the set of continuous unitary projective representations of $G$. We denote by $\text{CUP}_G$ the abelian group of equivalence classes of CUP-reps with the group operation
    $$[\mathbf U] \cdot [\mathbf V] := [\mathbf U \otimes \mathbf V].$$
    The following classification Theorem holds:

    \begin{proposition} \label{thm.isomorphism.CUP.and.H2}
    There is a group isomorphism between $\text{CUP}_G$ and $H^2_{\text{Borel}} (G,U(1))$.
    \end{proposition}

   \noindent For the proof, see Appendix \ref{appendix.cups}.
    
\subsection{SRE states and CUP representations}\label{sec: sre and cup}

We defer the full construction of the SPT index to Section \ref{sec.indexforGstates}. For now, we construct a \text{CUP}-rep starting from an SRE $G$-state, that is, a $G$-state that is strictly SRE rather than stably SRE.  This has been done by some authors before for discrete groups \cite{kapustin2021classification,ogata2021}, but some 
care is needed for the extension to continuous groups, e.g.\ as in Lemma \ref{lem: continuity of projective action}. For split states, an index for compact Lie groups was defined in \cite{Ragone:2024wie}, and it coincides with our definition. We make use of standard concepts like Galfand-Naimark-Segal (GNS) triple and we refer to \cite{bratteliI,bratteliII,naaijkens2017quantum} for their precise definitions and discussion.

\subsubsection{Construction of a CUP-rep} \label{sec.constructingcupreps}

    
    We partition the chain in half-chains  $\mathbb{Z}=L \cup R$ with 
     $L:= \{ j <0\},  R= \{ j \geq 0\}$. 
     Let $(\psi,\AA)$ be an $f$-SRE state with $\psi=\phi \circ \alpha_H$, with $\phi$ a product state and $\alpha_{H}$ an LGA.
     The product state $\phi$ has a GNS-triple of the form
    $$(\HH_{L,\phi} \otimes \HH_{R,\phi}, \pi_{L,\phi} \otimes \pi_{R,\phi}, \Omega_{L}\otimes \Omega_{R}),$$
where $(\HH_{L,\phi},\pi_{L,\phi}, \Omega_{L,\phi})$ and $(\HH_{R,\phi},\pi_{R,\phi}, \Omega_{R,\phi})$ are GNS triples of the restrictions $\phi|_L$ and $\phi|_R$, respectively. By a short calculation, it follows then that 
    \begin{align} 
\nonumber        &(\HH_{L}\otimes \HH_{R}, \pi_{L} \otimes \pi_{R}, \Omega) \\
\label{eq.splitgns}        &:=(\HH_{L,\phi}\otimes \HH_{R,\phi}, \pi_{L,\phi} \circ \alpha_{H_L} \otimes \pi_{R,\phi}\circ \alpha_{H_R}, \pi_{L}\otimes\pi_{R}(W^*) \Omega_{L}\otimes \Omega_{R})
    \end{align}
    is a GNS triple for $\psi$, where $\pi_{L}, \pi_{R}$ are again representations of $\caA_{L},\caA_R$ and 
    \begin{enumerate}
        \item $\alpha_{H_L}$ is an LGA on $\caA_L$ obtained by restricting $H$ to the halfline $L$, i.e.\ setting $H_L(S)=0$ whenever $S \not\subset L$.  Similarly for $\alpha_{H_R}$.
        \item   $W\in \AA$ is a unitary satisfying (see Theorem \ref{thm.lgadecomposition})
    \begin{equation*}
        \alpha_H = \alpha_{H_L}\otimes \alpha_{H_R} \circ \Ad{W}.
    \end{equation*}
    \end{enumerate}
    The existence of a factorized GNS representation as in \eqref{eq.splitgns} is known as the \emph{split property}.  The above remarks therefore show that  SRE states enjoy the split property.  
    Since the states $\phi_{\Gamma} \circ \alpha_{H_\Gamma}$ with $\Gamma = L,R$, are pure, the representations $(\HH_{\Gamma},\pi_{\Gamma},\Omega_\Gamma)$ are irreducible \cite{bratteliI}, and $\pi_\Gamma (\AA_\Gamma)'' = \mathcal B(\HH_\Gamma)$. 
    The following two lemmas are standard, see e.g. \cite{bratteliII}, and \cite{ogata2021,ogata2019classification}.

    \begin{lemma} \label{lem.strongcontinuity.unitaryimplementation}
        Let $(\psi,\AA,U)$ be an SRE $G$-state with GNS triple $(\HH_\psi,\pi_\psi,{\Omega_\psi})$. Then there exists a strongly continuous unitary representation $(\HH_\psi,u)$ of $G$, such that $u(g){\Omega_\psi} = {\Omega_\psi}$ and
        \begin{equation*}
		\pi_\psi \circ \beta_g = \Ad{u(g)} \circ \pi_\psi,\qquad g \in G.
	\end{equation*}
    \end{lemma}
 
	\begin{lemma} \label{lem.splitunitary}
   Let $(\HH_\Gamma,\pi_\Gamma)$, $(\HH'_\Gamma,\pi'_\Gamma)$ be irreducible representations of $\AA_\Gamma$, for $\Gamma =L,R$. If there is a unitary $w\in \UU(\HH_{L} \otimes \HH_{R},\HH'_{L} \otimes \HH'_{R} )$ such that
		\begin{equation}
			\Ad{w} \circ \pi_{L} \otimes \pi_{R} = \pi'_{L}  \otimes \pi'_{R},  
		\end{equation}
		then there exist unitaries $w^{\Gamma} \in \UU(\HH_{\Gamma},\HH'_{\Gamma})$ such that 
		\begin{equation}
			w=w^L\otimes w^R,\qquad \pi'_{\Gamma}  = \Ad{w^{\Gamma}} \circ \pi_{\Gamma}.  
		\end{equation}
	\end{lemma}

  By Lemma \ref{lem.strongcontinuity.unitaryimplementation},   the G-invariance of an SRE $G$-state $(\psi,\AA,U)$ implies the existence of unitaries $u(g)$ implementing the symmetry. The group action on $\AA$ can be split between left and right, since it acts as a product on $\AA=\AA_L\otimes\AA_R$, i.e.\  $\beta_g=\beta_g^L \otimes \beta_g^R$. Therefore, we have
    \begin{equation*}
        \pi_L \otimes \pi_R \circ \beta_g^L \otimes \beta_g^R = \Ad{ u(g) } \circ \pi_L \otimes \pi_R. 
    \end{equation*}
    Hence, Lemma \ref{lem.splitunitary} gives  a (non-unique) factorization
    \begin{equation} \label{eq.splittinggroup}
        u(g) = u^L(g) \otimes u^R(g),
    \end{equation}
    and, consequently,
    \begin{equation*}
        \pi_\Gamma \circ \beta_g^\Gamma= \Ad{u^\Gamma(g)} \circ \pi_\Gamma,
    \end{equation*}
    for $u^\Gamma(g) \in \mathcal U(\HH_\Gamma)$. \\

The following lemma does not appear in the literature dealing with finite groups $G$, but it is crucial for us.
    \begin{lemma}\label{lem: continuity of projective action}
        Let $\mathscr P:\ \UU(\HH_R) \to \mathbb P(\UU(\HH_R))$ be as defined in section \ref{sec.cups}. Given a splitting of the unitary group representation as in equation \eqref{eq.splittinggroup}, the map 
        \begin{equation} \label{eq.cuprepresentation}
            \mathscr P \circ u^R:\ G \to \mathbb P(\UU(\HH_R))
        \end{equation}
        is a continuous unitary projective representation of $G$.
    \end{lemma}
    \begin{proof}Recall that the quotient topology in $\mathbb P(\HH_R)$ is metrizable, and we can write
        \begin{equation*}
            d(\mathbf{a},\mathbf{b}) = \inf\limits_{a\in \mathbf {a},\ b \in \mathbf {b}} \norm{\dfrac{a}{||a||} - \dfrac {b}{ ||b||}} = \left[ 2(1-|\langle a,b\rangle|) \right]^{\frac 12},
        \end{equation*}
        where the right-hand side does not depend on representatives $a,b$ of $\mathbf a, \mathbf b$. Let $\mathbf u^R := \mathscr P \circ u^R$. Then
        \begin{align} \label{eq.distrays}
            \dfrac 12 d(\mathbf u^R(g)\mathbf a, \mathbf a)^2 = 1 - |\langle u^R(g)a,a \rangle|, \qquad \mathbf a \in \mathbb P(\HH_R). 
        \end{align}

        By assumption of continuity of the on-site group action, it holds that $g \mapsto \Ad{u(g)} \left(\pi_\psi (A)\right)$ is continuous for each $A \in \AA$, and that continuity extends to
        \begin{equation} \label{eq: extended cont}
           g \mapsto \Ad{u(g)}(B), 
        \end{equation}
         for any  $B \in \mathcal B(\HH_L\otimes \HH_R)$.
         Indeed, the continuity extends to the strong closure of $\pi_\psi (\AA)$ in $\mathcal B(\HH_L\otimes \HH_R)$. Since 
     $\psi$ is a pure state, the GNS representation is irreducible, and this strong closure equals  $ \mathcal B(\HH_L\otimes \HH_R)$. 
     \par 
        We choose now vectors ${\Omega_R} \in \HH_R$ and $\Omega_L \in \mathcal H_L$, and let $a = \pi_R(A_R){\Omega_R} \in \HH_R$, for $A_R \in \AA_R$. Then
\begin{align*}
            \vert\langle  a, u^R(g) a \rangle\vert^2 &= \langle a, \Ad{u^R(g)}  (\ketbra{a}) a \rangle \\ &= \langle  a \otimes \Omega_L,  \Ad{u(g)} \ketbra{a}\otimes \mathds 1_{\HH_L}   
            a \otimes \Omega_L \rangle
        \end{align*}
 and by continuity of \eqref{eq: extended cont}, it follows that the map
        $$g\mapsto \vert\langle a,u^R(g)a \rangle\vert$$
        is continuous. Since  $\pi_R(\caA_R){\Omega_R}$ is dense in $\mathcal H_R$ (by irreducibility of the representation $(\HH_R,\pi_R)$), the continuity extends to any $a\in \mathcal H_R$. Along with \eqref{eq.distrays}, this proves continuity of the unitary projective representation $\mathbf u^R = \mathscr P \circ u^R$.

    \end{proof}

    \subsubsection{An index for SRE G-states}

    \noindent In the previous section, we explained how to construct a CUP-rep $\mathbf u^R$ from a SRE $G$-state $(\psi,\mathcal A, U)$, and how to associate an element of $H^2$ (which we henceforth use as an abbreviation for $H_{\text{Borel}}^2 (G,U(1))$) to a $\mathbf u^R$. The construction of $\mathbf u^R$ depends on the choice of a split GNS representation (via the choice of the LGA $\alpha_H$ and the product state $\phi$). However, we now prove that the associated element of  $H^2$ does not depend on these choices:
    
    \begin{proposition} \label{thm.index}
    The element of $H^2$ obtained as discussed above, is a well-defined function of the SRE G-state $(\psi,\mathcal A, U)$. Moreover, it is stable with respect to $G$-equivariant LGA's. 
    
    \end{proposition}
    \begin{proof}
        Let $(\HH_L \otimes \HH_R, \pi_L \otimes \pi_R, \Omega)$ and $(\HH_L' \otimes \HH_R', \pi_L' \otimes \pi_R',{\Omega'})$ be two split GNS triples associated to the SRE G-state $\psi$. By uniqueness of GNS representation, there exists a unitary $W: \HH_L\otimes \HH_R \to \HH_L'\otimes \HH_R'$ such that $\Ad{W} (\pi_L\otimes \pi_R) = \pi_L' \otimes \pi_R'$ and $W \Omega =  \Omega'$. As a consequence of Lemma \ref{lem.splitunitary}, it splits as
        $$W = W_L \otimes W_R, \qquad W_\Gamma \in \UU(\HH_\Gamma,\HH_\Gamma'),\ \Gamma = L,R.$$
        \par 
        Let $\mathscr P \circ u^R$ and $\mathscr P \circ u^{R'}$ be the $\text{CUP}_G$ obtained from the split GNS triples above. It holds
        \begin{align*}
            \pi_R'\circ \beta_g &= \Ad{u^{R'}(g)} \circ \pi_R'\\
            &= \Ad {W_R} \circ \pi_R \circ \beta_g \\
            &= \Ad {W_R u^R(g) W_R^*} \circ \pi'_R.
        \end{align*}
        It follows that $u^{R'} (g) W_R u^R(g) W_R^* \in \pi_R(\AA_R)' = \CC \mathds 1_R$. Hence
        $$u^{R'} (g) = c(g) W_R u^R(g) W_R^*, \qquad c(g) \in U(1),$$
        which implies that $\mathscr P \circ u^{R'}$ is equivalent (see Definition \ref{def.equivalenceofcups}) to $\mathscr P \circ u^R$. \\ 

        For the second part of the proof, let $\alpha_{H}$ be a G-equivariant LGA. If $(\HH_L \otimes \HH_R, \pi_L \otimes \pi_R,  \Omega)$ is a split GNS triple of $\psi$, the SRE $G$-state $(\psi \circ \alpha_H,\AA,U)$ has a split GNS triple of the form
        $$(\HH_L \otimes \HH_R, \pi_L \circ \alpha_{H_L} \otimes \pi_R \circ \alpha_{H_R}, \pi_L \otimes \pi_R (W^*) \Omega),$$
        where $\Ad{W} = \alpha_H \circ (\alpha_{H_L} \otimes \alpha_{H_R})^{-1}$, as in Theorem \ref{thm.lgadecomposition}. By using the $G$-equivariance of $\alpha_{H_R}$, we see that
        \begin{align*}
            \pi_R \circ \alpha_{H_R} \circ \beta_g &= \pi_R \circ \beta_g \circ \alpha_{H_R} \\
            &= \Ad{u^R(g)} \circ \pi_R \circ \alpha_{H_R},
        \end{align*}
        from which we conclude the $\text{CUP}_G$ remains the same, since the group can be represented by the same unitary action $u^R$.
    \end{proof}

    \spacing
We are now finally ready to  define the SPT-index for SRE $G$-states. 
    Let $(\psi,\mathcal A, U)$ be a SRE $G$-state. Its index $\sigma^{(\text{SRE})}_{(\psi,\mathcal A,U)}$ is set to be the element of $H^2$ corresponding (cf.\ Proposition \ref{thm.isomorphism.CUP.and.H2}) to a CUP-rep \eqref{eq.cuprepresentation} obtained from $(\psi,\AA,U)$.  This is well-defined by virtue of Proposition \ref{thm.index}.

\subsubsection{Properties of the index for SRE states}

Given $j\in \mathbb Z$, we define infinite half-chains
    $$L^{(j)}:= (-\infty, j)\cap \ZZ,\qquad R^{(j)}:= [j,\infty)\cap \ZZ.$$
    For any choice of $j$, we could perform the same construction as in Section \ref{sec: sre and cup}, replacing $R,L$ by $R^{(j)},L^{(j)}$, and obtain an index in $H^2_{\text{Borel}} (G,U(1))$. We now show that we would get the same index in this way. 
    \begin{lemma}
        The index does not depend on the choice of $j\in \bbZ$.
    \end{lemma}
    
    \begin{proof}
        Let $i<j$ be two sites. The GNS representations $(\HH_{L^{(i)}}\otimes \HH_{R^{(i)}}, \pi_{L^{(i)}}\otimes \pi_{R^{(i)}})$ and $(\HH_{L^{(j)}}\otimes \HH_{R^{(j)}}, \pi_{L^{(j)}}\otimes \pi_{R^{(j)}})$ of $\psi$ are unitarily equivalent, via a unitary $V \in \UU(\HH_{L^{(i)}}\otimes \HH_{R^{(i)}},\HH_{L^{(j)}}\otimes \HH_{R^{(j)}})$, by uniqueness of the GNS representation. On the one hand, 
        \begin{align}
        \nonumber    \pi_{R^{(j)}} \circ \beta_g^{R^{(j)}} &= \Ad{u^{R^{(j)}} (g)} \circ \pi_{R^{(j)}} \\
        \label{eq.index-cut1}    &= \Ad{u^{R^{(j)}}(g)} \circ \Ad V \circ \pi_{R^{(i)}} .
        \end{align}
        On the other hand, 
        \begin{align} 
        \nonumber    \pi_{R^{(j)}} \circ \beta_g^{R^{(j)}}  &= \Ad V \circ \pi_{R^{(i)}} \circ \beta^{R^{(i)}}_g \circ \beta^{R^{(i)} \setminus R^{(j)}}_{g^{-1}}\\
        \label{eq.index-cut2} &= \Ad V \circ \Ad{u^{R^{(i)}}(g)} \circ \pi_{R^{(i)}} \circ \beta^{R^{(i)} \setminus R^{(j)}}_{g^{-1}}.
        \end{align}
        Define, for each region $S \in \mathcal P_{\text{fin}}(\ZZ)$, 
        $$U_S(g) := \bigotimes\limits_{j \in S} U_j(g),\qquad g \in G,$$
        where $U_j(g)$ is the on-site unitary group action. For an $h \in G$, we know $\beta_{h}^{R^{(i)}\setminus R^{(j)}} = \Ad {U_{[i,j)}(h)}$. It follows that $W_g := \pi_{R^{(i)}} (U_{[i,j)}(g))$ is a unitary representation. Then, as a consequence of \eqref{eq.index-cut1} and \eqref{eq.index-cut2}, it holds that
        \begin{align*}
            \Ad{u^{R^{(j)}}(g)} \circ \Ad V \circ \pi_{R^{(i)}} (A) &=\Ad V \circ \Ad{u^{R^{(i)}}(g)} \circ \pi_{R^{(i)}} \circ \beta^{R^{(i)} \setminus R^{(j)}}_{g^{-1}} (A) \\
            &= \Ad V \circ \Ad{u^{R^{(i)}}(g)} \circ \Ad{W_g} \circ \pi_{R^{(i)}} (A),
        \end{align*}
        for each $A \in \AA_{R^{(i)}}$. Hence
        \begin{align*}
            \Ad{ V^* u^{R^{(j)}} (g^{-1}) V u^{R^{(i)}}(g) W_g } \circ \pi_{R^{(i)}} (A) = \pi_{R^{(i)}} (A),
        \end{align*}
        which implies
        \begin{align} \label{eq.siteindependentindex}
            V^* u^{R^{(j)}} (g^{-1}) V u^{R^{(i)}}(g) W_g  \in \pi_{R^{(j)}}(\AA_{R^{(j)}})' = \CC \mathds 1,
        \end{align}
        where the equality follows from irreducibility of the GNS representation. Furthermore, it is clear that $[V^*u^{R^{(j)}}(g')V, W_g] = 0$ for every $g,g' \in G$. Let $\mu$ be a Borel multiplier associated to $u^{R^{(j)}}$. Then
        \begin{align*}
            V^* u^{R^{(j)}} (g) V W_g V^* u^{R^{(j)}} (h) V W_h u^{R^{(j)}} (h) &\left[ V^* u^{R^{(j)}} (gh) V W_{gh} \right]^{-1} \\ &= V^* u^{R^{(j)}} (g) u^{R^{(j)}} (h) \left[ u^{R^{(j)}}(gh)  \right]^{-1}V W_gW_hW_{gh}^{-1} \\ &= \mu(g,h) \cdot \mathds 1,
        \end{align*}
        since $W_g$ is a linear representation. Hence $V^*u^{R^{(j)}}V W$ is a Borel lift of the same CUP-rep as $u^{R^{(j)}}$. By equation \eqref{eq.siteindependentindex}, it holds that 
        \begin{align} \label{eq.index-cup-cohomologous}
            u^{R^{(i)}}(g) = \nu(g) V^* u^{R^{(j)}} (g) V W_g^*, \qquad g \in G, 
        \end{align}
        for some measurable map $\nu: G \to U(1)$. Hence, if $\mu^{(i)}, \mu^{(j)}$ are Borel cocycles associated to $u^{R^{(i)}}$ and $u^{R^{(j)}}$, then equation \eqref{eq.index-cup-cohomologous} implies $\mu^{(i)}(g,h) = \nu(g) \nu(h) \nu^{-1}(gh) \mu^{(j)}(g,h)$. The cocycles are cohomologous and, consequently, correspond to the same element in $H^2$. 
        
    \end{proof}

    \spacing 

    \begin{lemma} \label{lem.stackingindex}
        The index respects the stacking operation. Namely, given two SRE $G$-states $(\psi, \AA, U)$ and $(\psi', \AA', U')$, we have
        \begin{align*}
            \sigma^{(\text{SRE})}_{(\psi \stack \psi', \AA \stack \AA', U \stack U')} = \sigma^{\text{(SRE)}}_{(\psi, \AA, U)} \cdot \sigma^{\text{(SRE)}}_{(\psi',\AA', U')}.
        \end{align*}
    
    \end{lemma}
    \begin{proof}
        The SRE $G$-state 
        $(\psi \stack \psi', \AA \stack \AA', U \stack U')$
        has a GNS triple $(\HH_\psi \stack \HH_{\psi'}, \pi_\psi \stack \pi_{\psi'}, \Omega_\psi \stack \Omega_{\psi'}).$ It is clear that 
        $(\mathscr P \circ u^R) \stack ( {\mathscr P}' \circ (u')^R)$
        is a CUP-rep (as in Lemma \ref{lem: continuity of projective action}) representing the stacked group action $\beta_g^R \stack (\beta_g')^R$. From the discussion in Section \ref{sec.equivalenceofprojectivereps}, it follows that
        \begin{align*}
            \left[ (\mathscr P \circ u^R) \stack ( {\mathscr P}' \circ (u')^R)  \right] = \left[ (\mathscr P \circ u^R) \right] \cdot \left[ ( {\mathscr P}' \circ (u')^R) \right],
        \end{align*}
        which proves the result. 
    \end{proof}

\subsection{An index for G-states} \label{sec.indexforGstates}
    In the previous subsections we explained how to obtain an $H^2$-valued index for a SRE $G$-state, and we proved some of its properties. Let $(\psi,\AA,U)$ be a $G$-state that is not necessarily strictly SRE. There exists a product state $(\phi,\AA_\phi)$ for which $(\psi \stack \phi, \AA \stack \AA_\phi)$ is SRE. We equip $(\phi,\AA_\phi)$ with a trivial group action $U_\phi \equiv \mathds 1$, and define the index
    \begin{align} \label{eq.definitionindexGstate}
        \sigma_{(\psi, \AA, U)} := \sigma^{(\text{SRE})}_{(\psi \stack \phi, \AA \stack \AA_\phi, U \stack U_\phi)}
    \end{align}
    of the $G$-state $(\psi,\AA,U)$ as the index of the SRE $G$-state $(\psi \stack \phi, \AA \stack \AA_\phi, U \stack U_\phi)$. 
    
    \spacing

    \begin{proposition} \label{thm.indexhomomorphism}
    The index map defined above restricts to a group homomorphism
    \begin{equation} \label{eq.indexmap}
        \sigma:\ \text{SPT}_G \to H^2_{\text{Borel}} (G,U(1)).
    \end{equation}

    \end{proposition}
    \begin{proof}
        We first prove that this map is well-defined. Firstly, we prove it does not depend on the choice of stacked $G$-product state $(\phi, \AA_\phi, \mathds 1)$ considered in equation \eqref{eq.definitionindexGstate}. Choose a different product state $( \phi',\AA_{\phi'})$ such that $(\psi \stack \phi', \AA \stack \AA_{\phi'} , U \stack U_{\phi'})$ is G-invariant SRE, where $U_{ \phi'}$ is again a trivial group action. It holds that:
    \begin{enumerate} [label=\roman*)]
        \item The indices of $(\phi, \AA_\phi, U_{\phi})$ and $(\phi', \AA_{\phi'}, U_{\phi'})$ are trivial,

        \item SRE $G$-states 
        $$(\psi \stack \phi \stack \phi', \AA \stack \AA_\phi \stack \AA_{\phi'}, U \stack U_\phi \stack U_{ \phi'})$$ 
        and 
        $$(\psi \stack \phi' \stack \phi, \AA \stack \AA_{\phi'} \stack \AA_\phi, U \stack U_{\phi'} \stack U_\phi)$$
        are $G$-equivalent, hence, by Proposition \ref{thm.index}, have the same SRE index. 
    \end{enumerate}
    Consequently, by Lemma \ref{lem.stackingindex}, it holds
    \begin{align*}
        \sigma_{(\psi \stack \phi, \AA \stack \AA_\phi, U \stack U_\phi)} &= \sigma_{(\psi \stack \phi, \AA \stack \AA_\phi, U \stack U_\phi)} \cdot \sigma_{(\phi, \AA_\phi, U_\phi)} 
        \\ &= \sigma_{(\psi \stack \phi \stack \phi', \AA \stack \AA_\phi \stack \AA_{ \phi'}, U \stack U_\phi \stack U_{\phi'})} \\ &= \sigma_{(\psi \stack \phi' \stack \phi, \AA \stack \AA_{\phi'} \stack \AA_\phi, U \stack U_{\phi'} \stack U_\phi)} \\ &=
        \sigma_{(\psi \stack \phi', \AA \stack \AA_{ \phi'} , U \stack U_{\phi'})} \cdot \sigma_{(\phi', \AA_{\phi'},U_{\phi'})} \\
        &= \sigma_{(\psi \stack \phi', \AA \stack \AA_{\phi'} , U \stack U_{\phi'})}.
    \end{align*}
    Hence the index does not depend on the choice of stacked $G$-product state. \par
    Secondly, we prove the index does not depend on a representative of a class in $\text{SPT}_G$: let $(\psi,\AA,U)$ and $(\psi',\AA',U')$ be $G$-stably equivalent. There are $G$-product states $(\varphi, \AA_\varphi, U_\varphi)$ and $(\varphi',\AA_{\varphi '},U_{\varphi'})$ such that $(\psi \stack \varphi,\AA \stack \AA_\varphi, U \stack U_\varphi)$ and $(\psi' \stack \varphi',\AA' \stack \AA_{\varphi'}, U'\stack U_{\varphi '})$ are $G$-equivalent. On the other hand, there are $G$-product states $(\phi, \AA_\phi, \mathds 1)$ and $(\phi',\AA',\mathds 1)$ such that $(\psi \stack \phi , \AA\stack \AA_\phi, U \stack \mathds 1)$ and $(\psi'\stack \phi', \AA\stack \AA_{\phi'}, U'\stack \mathds 1)$ are SRE $G$-states. Thus, the SRE $G$-states
    $$(\psi \stack \varphi \stack \phi \stack \phi') \qquad \text{and} \qquad (\psi' \stack \varphi' \stack \phi \stack \phi')$$
    are $G$-equivalent. By Proposition \ref{thm.index}, they have the same index. Since $G$-product states have trivial index, Lemma \ref{lem.stackingindex} allows us to conclude
    $$\sigma_{(\psi,\AA,U)} = \sigma^{\text{SRE}}_{(\psi\stack \phi, \AA \stack \AA_\phi, U\stack U_\phi)} = \sigma^{\text{SRE}}_{(\psi'\stack \phi', \AA' \stack \AA_{\phi'}, U'\stack U_{\phi'})} = \sigma_{(\psi',\AA',U')}.$$
    Hence the index for $G$-states is well-defined. It also follows from Lemma \ref{lem.stackingindex} and the first part of this proof that the index map for $G$-states respects the stacking operation. 
    \end{proof}
    
    

    
    \spacing 

    Next, we show that every element of $H^2$ corresponds to some $G$-state: 
    \begin{lemma} \label{lem.surjective}
        The index map \eqref{eq.indexmap} is surjective.
    \end{lemma}
    \begin{proof}
    Given a 2-cocycle $\mu$, representative of an element $[\mu] \in H^2$, one can construct a finite-dimensional irreducible CUP-rep $(\mathcal V,v)$ of $G$ with $\mu$ as a Borel multiplier, that is, a pair $(\mathcal V,v)$ consisting of a finite-dimensional Hilbert space $\mathcal V$ and maps $\mathcal:\ g \to \UU(\mathcal V)$ satisfying
        \begin{equation*}
            v(g)v(h) v(gh)^{-1} = \mu(g,h),\qquad g,h \in G.
        \end{equation*} 
    In fact, irreducible CUP-reps of $G$ correspond uniquely to certain irreducible unitary representations of the group extension $G^{(\mu)}$ of $G$ by $U(1)$ \cite{mackeyI}. Given a suitable topology, $G^{(\mu)}$ is compact, and by a corollary of Peter Weyl's theorem, see \cite{johnson}, $\mathcal V$ is finite-dimensional. \par 

    We construct a triple $(\psi,\AA,U)$ with index $[\mu] \in H^2_{\text{Borel}}(G,U(1))$. Define on-site Hilbert spaces $\HH_j := \mathcal V\otimes \overline {\mathcal V}$, where $(\overline{\mathcal V}, \overline v)$ is the complex conjugate multiplier representation of $(\mathcal V,v)$. The unitary group action is then the linear representation $g \mapsto v(g) \otimes \overline{v(g)}$. The chain algebra $\AA$ is constructed from the on-site algebras $\AA_j = \mathcal B(\mathcal V\otimes \overline {\mathcal V})$. \par 
    There is a trivial representation $t_{j,j+1}\in \overline {\mathcal V}_j \otimes {\mathcal V}_{j+1}$, see the proof of Part (ii) of Proposition \ref{thm.singlets}. Let $\psi$ be the entangled pair state obtained formally as 
    $$\bigotimes\limits_{j\in \ZZ} \langle t_{j,j+1},(\ \cdot\ ) t_{j,j+1}\rangle.$$
    By construction it is a $G$-state. We claim the index of $(\psi, \AA, v \otimes \overline v)$ is $[\mu]$. In fact, when restricted to $R^{(0)}$, the group acts formallly as 
    $$v(g) \otimes (\overline {v(g)} \otimes v(g)) \otimes (\overline {v(g)} \otimes v(g)) \otimes \dots,$$
    which, in effect, corresponds to the projective representation $(\mathcal V,v)$, with index $[\mu]$.
    \end{proof}

    \spacing

    \begin{remark} [Triviality of the index] \label{remark.triviality}
    As already pointed out in section \ref{sec.equivalenceofprojectivereps}, a CUP-rep $\mathbf U$ of $G$ has trivial $H^2$ index iff.\ it has a lift $u$ that is a strongly continuous unitary representation of $G$ (Remark 2' of \cite{cattaneo}). Therefore, the index of a class $[(\psi,\AA,U)] \in \text{SPT}_G$ is trivial iff.\ the unitary action $u^R$ can be chosen to be a strongly continuous unitary representation of $G$ on $\HH_R$.
    \end{remark}

	
	\subsection{Completeness of Classification}

    We state our result regarding completeness of the classification of 1d SPT's and we use it to give a proof of Theorem \ref{thm.grouphomomorphism}.
    
	\begin{theorem} \label{thm.main}
		Let $(\psi,\AA,U)$ be a G-state. If its index is trivial, i.e.\ 
        $$\sigma_{(\psi,\AA,U)} = [1],$$
        then $(\psi,\AA,U)$ is G-stably equivalent (see Definition \ref{def.gequivalence}) to a special G-product state.
	\end{theorem}

    \spacing

    \noindent Theorem \ref{thm.main} is proved in Section \ref{sec.proofmain}, since it requires a sequence of technical preliminary results that we discuss carefully over the rest of the paper. 

So far, we have proved already in Lemma \ref{lem.surjective} that the index map is surjective. By Theorem \ref{thm.main}, we obtain the injectivity. 
    \begin{corollary} \label{cor.injective}
        The index map \eqref{eq.indexmap}
        is injective.
    \end{corollary}
    \begin{proof}
    Let $(\psi,\AA,U),\ (\psi',\AA',U')$ have the same index. By surjectivity of $\sigma$, there is a corresponding $(\overline{\psi},\overline A, \overline U)$ with inverse index, and, by Theorem \ref{thm.main}, both
    \begin{equation*}
    (\psi'\tilde \otimes \overline{\psi}, \AA' \tilde\otimes \overline {\AA}, U\tilde\otimes \overline U), 
    \qquad (\psi \tilde \otimes \overline{\psi}, \AA \tilde\otimes \overline {\AA}, U\tilde\otimes \overline U)
    \end{equation*}
    are in the trivial phase, i.e. are G-stably equivalent to special G-product states, since the index respects the stacking operation. It follows that the $G$-state
    $$( \psi'\stack \overline \psi \stack \psi, \AA' \stack \overline{\AA} \stack \AA , U' \stack \overline U \stack U)$$
    is $G$-stably equivalent to both $(\psi,\AA,U)$ and $(\psi',\AA',U')$. Hence $\psi$ is G-stably equivalent to $\psi'$, by transitivity.
    \end{proof}
    \vspace{.6cm}

    \noindent We can therefore conclude the proof of the main result: 
    
    \begin{proof} [Proof of Theorem \ref{thm.grouphomomorphism}] \label{section.proofhomomorphism}

    From Section \ref{sec.indexforGstates}, there is a well defined index associated to each phase in $\text{SPT}_G$. Furthermore, the index map 
    $$\sigma:\ [(\psi,\AA,U)] \mapsto \sigma_{(\psi,\AA,U)}$$
    is a bijective (Lemma \ref{lem.surjective} and Corollary \ref{cor.injective}) homomorphism (Proposition \ref{thm.indexhomomorphism}).
    \end{proof}


\section{Technical Preliminaries} \label{sec.technical}
We state a few general results that will be used in Section \ref{sec.proofmain}. 
\subsection{G-product states versus special G-product states} \label{appendix.ginvariantproductstates}
    In this subsection we prove that the requirement of stacking with a special G-invariant product state in the definition of G-stable equivalence can be weakened: the G-invariant product states one can stack with are allowed to have non-zero on-site G-charges (see Section \ref{sec:productstates}). This is because any two $G$-product states are G-stably equivalent, as it will be proved in Theorem \ref{thm.equivalenceofproductstates}. This result is inspired by \cite{kapustin2021classification} but the proof is different because of the requirement of bounded on-site Hilbert space dimension. 

We need the following:
    \begin{lemma} \label{lem.specialgstate}
        Let $(\psi,\AA,U)$ be a $G$-state. Then it is G-stably equivalent to a $G$-state \linebreak
        $\left(\psi \tilde \otimes \phi, \AA \tilde \otimes \AA', U\tilde\otimes U'\right)$
        such that there exists a special $G$-product state
        $\left(\phi', \AA \stack \AA', U \stack U' \right).$
    \end{lemma}
    \begin{proof}
        For each $j \in \ZZ$, denote by $(\overline{\HH_j}, \overline{U_j})$ the complex conjugate representation to $(\HH_j,U_j)$. It has the same dimension as $\HH_j$. We construct an auxiliary chain algebra $\AA'$ with on-site Hilbert spaces decomposing as
        $$\HH'_j = \mathbb Ce_j \oplus \overline{\HH_j},$$
        for a one-dimensional representation space $\CC e_j$, with respect to the group action
        $$U_j' = \mathds 1 \oplus \overline{U_j}.$$
        The $G$-state defined as
        $$\left(\psi \tilde \otimes \phi, \AA \tilde \otimes \AA', U\tilde\otimes U'\right):= \left(\psi \tilde \otimes \left( \bigotimes_{j\in \ZZ} \langle e_j, (\cdot) e_j \rangle \right), \AA \tilde \otimes \AA', U\tilde\otimes U'\right)$$
        is clearly G-stably equivalent to $(\psi,\AA,U)$. Each on-site Hilbert space $\HH_j\tilde \otimes \HH_j'$ contains now a singlet $\xi_j \in \HH_j \otimes \overline{\HH_j}$ \cite{creutz,fuchs}, giving rise to a special $G$-product state 
        $$\left(\phi', \AA \stack \AA', U \stack U' \right):=\left( \bigotimes \limits_{j \in \ZZ} \langle \xi_j, (\cdot) \xi_j \rangle, \AA \stack \AA', U \stack U' \right).$$
    \end{proof}

	\begin{theorem} \label{thm.equivalenceofproductstates}
		Any two $G$-product states are $G$-stably equivalent.  
	\end{theorem}
	\begin{proof}
	    Let $(\phi,\AA,U)$ be a $G$-product state, characterized by on-site one-dimensional Hilbert spaces $\CC \nu_i$, transforming as continuous $q_i \in \hom({G, U(1)})$. Our goal is to prove $(\phi,\AA,U)$ is $G$-stably equivalent to a special $G$-product state, which will imply the theorem. Without loss of generality, by Lemma \ref{lem.specialgstate}, there is a special $G$-product state
        $$\left( \bigotimes\limits_{i \in \ZZ} \langle e_i, (\cdot ) e_i\rangle, \AA, U \right).$$
     
        Assume first that $q_i = 1$ for every $i<0$. Consider a chain algebra $\AA'$ carrying on-site representations $(\HH_i',U_i')$ that decompose as direct sums
        $$\CC e_i' \oplus \CC \nu_i' \oplus \CC \xi_i' = \HH_i',$$
        where
        \begin{enumerate}
            \item $(\CC e_i',U_i')$ is a trivial representation;
            \item $(\CC\nu_i',U_i')$ is a representation transforming as $q_i$;
            \item $(\CC\xi_i',U_i')$ are representations such that (for $i=0,1,2,\dots$): 
		\begin{itemize}
			\item $\xi_0'$ transforms as $q_0$,
			\item $\xi_{2n+1}'$ transforms as $\overline q_{2n} \dots \overline q_1 \overline q_0$, where $\overline{q}$ denotes the conjugate representation to $q$,
			\item $\xi_{2n+2}'$ transforms as $q_{2n+2}\dots q_1 q_0$,
		\end{itemize}
        \end{enumerate}

        \spacing 

        We will now transform $(\phi \tilde \otimes \phi',\AA \stack \AA', U \stack U')$, where $\phi'$ is the special $G$-product state 
        $\otimes_i \langle e_i', (\cdot) e_i'  \rangle$,
        with on-site trivial representations $e_i'$, into a special $G$-product state, in 3 steps. \\
        
		\noindent \textbf{(a)} In each stacked on-site space $\HH_i \stack \HH_i'$, there is a $G$-invariant subspace $K_i$ spanned by $\nu_i \stack e_i'$ and $e_i \stack \nu_i'$, hence there exists a $G$-equivariant unitary $T_i$ swapping $\nu_i \stack e_i'$ and $e_i \stack \nu_i'$. For instance, for each $i \in \ZZ^+$, define
        \begin{equation}
            T_i := \left( \ket {\nu_i \stack e_i'} \bra{e_i \stack \nu_i'} + \ket {e_i \stack \nu_i'  } \bra{\nu_i \stack e_i'} \right) \oplus \mathds 1_{(K_i)^\perp}. 
        \end{equation}
        By the discussion in subsection \ref{sec.trivialTDI}, there is a G-equivariant TDI $H_T$ with finite $f$-norm, for any $f \in \mathcal F$, which implements the formal automorphism $\bigotimes_{i \in \ZZ^+} \Ad{T_i}$. To lighten the notation, we sometimes identify a vector with its corresponding one-dimensional representation (its $G$-charge). Hence, in summary, the family $\{T_i\}_i$ transfers non-trivial charges from the original chain algebra $\AA$ to the stacked chain algebra $\AA'$, i.e. it maps on-site charges $\{q_i \tilde \otimes 1\}_i$ to charges $\{1 \tilde \otimes q_i\}_i$, as in figure \ref{fig.productstate2}, producing a $G$-product state $(\phi\stack \phi' \circ \alpha_{H_T},\AA\stack \AA', U\stack U')$.
        \begin{figure} [h]
            \centering
            \includegraphics[scale=.25]{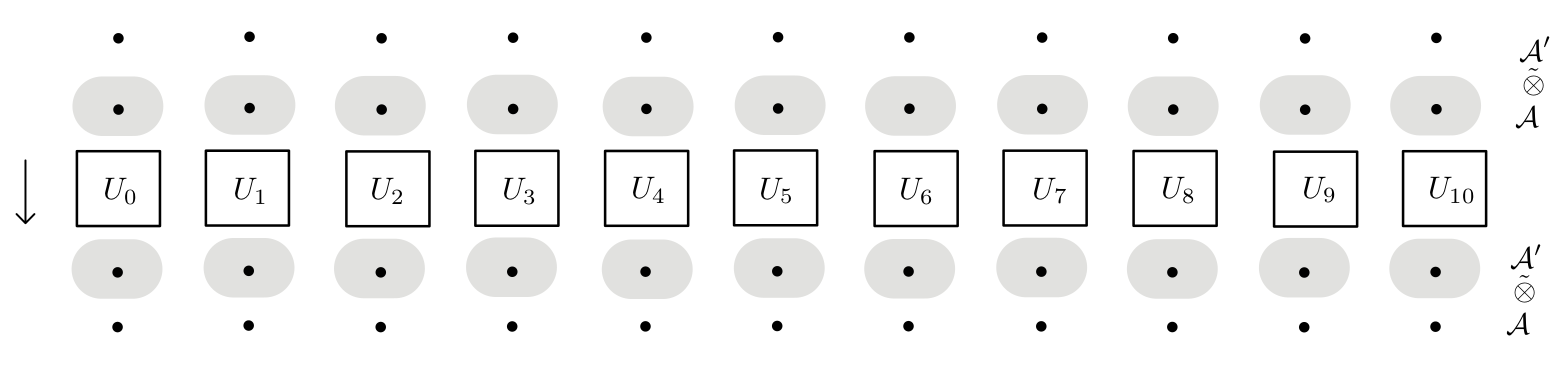}
            \caption{Non-trivial charges (in gray) are transferred from the original chain to the stacked chain.}
            \label{fig.productstate2}
        \end{figure}
        \spacing

        \noindent \textbf{(b)} Secondly, we construct another $G$-equivariant LGA rearranging the $G$-charges, as follows: 
        \begin{enumerate} [label=(\roman*)]
            \item Define $V_0 \in \UU(\HH_0 \stack \HH_0')$ as a $G$-equivariant unitary that restricts to the identity on $\HH_0$ and swaps between $\nu_0'$ and $\xi_0'$ on $\HH_0'$,
            

            \item For $j=2k+1$, $k=0,1,2,\dots$, define a unitary $V_j$ acting on $(\HH_j \stack \HH_j') \otimes (\HH_{j+1} \stack \HH_{j+1}')$ as the $G$-equivariant unitary that restricts to the identity on $\HH_j \otimes \HH_{j+1}$ and swaps between $\xi_j' \otimes \xi_{j+1}'$ and $\nu_j' \otimes \nu_{j+1}'$ on $\HH_j' \otimes \HH_{j+1}'$.
        \end{enumerate}
        
		\noindent For each $j=2k+1$, the corresponding unitary $V_{j}$ swaps charges $q_{2k+1}q_{2k+2}$ and $\overline q_{2k} \dots \overline q_0\cdot q_{2k+2} q_{2k+1} \dots q_0$. E.g. 
        \begin{itemize}
            \item $V_1$ maps $q_1 \otimes q_2$ into $\overline {q_0} \otimes q_2q_1q_0$,
            \item $V_3$ maps $q_3  \otimes q_4$ into $\overline {q_2 q_1 q_0}  \otimes q_4q_3q_2q_1q_0$,
            \item etc,
        \end{itemize}
        such that charges accumulate along the half-chain. Similarly as for the family of $T_i$'s, we can obtain a $G$-equivariant TDI $H_V$ generating the automorphism
        \begin{align*}
            \alpha_{H_V} = \bigotimes\limits_{k \in \ZZ^+} \Ad{V_{2k+1}}.
        \end{align*}
        
		\spacing 
  
        \noindent \textbf{(c)} Finally, after applying the unitaries $V_{2k+1}$, we end up with a state
        as represented in Figure \ref{fig.productstate3} a),
        having on-site charges
		\begin{align*}
			\dots \otimes (1 \tilde \otimes 1)_{\text{site -1}} \otimes (1 \tilde \otimes q_0)_{\text{site 0}} \otimes (1 \tilde \otimes \overline q_0) \otimes (1 \tilde \otimes q_2q_1q_0) \otimes (1 \tilde \otimes \overline{q_2q_1q_0}) \otimes (1 \tilde \otimes q_4q_3q_2q_1q_0) \otimes \dots 
		\end{align*}
		It is now clear that pairs with opposite charges can be canceled locally by a third family of $G$-equivariant 2-local unitaries $W_{2k}$, $k=0,1,2,3,\dots$. I.e., there is a $G$-equivariant LGA $\alpha_{H_{W}}$ such that
        \begin{equation}
            \left( \phi\stack \phi' \circ \alpha_{H_T} \circ \alpha_{H_V} \circ \alpha_{H_W}, \AA \stack \AA', U \stack U' \right)
        \end{equation}
        is a special $G$-product state, when restricted to the right of site $0$. 
        
        \begin{figure} [h]
            \centering
            \includegraphics[scale=.25]{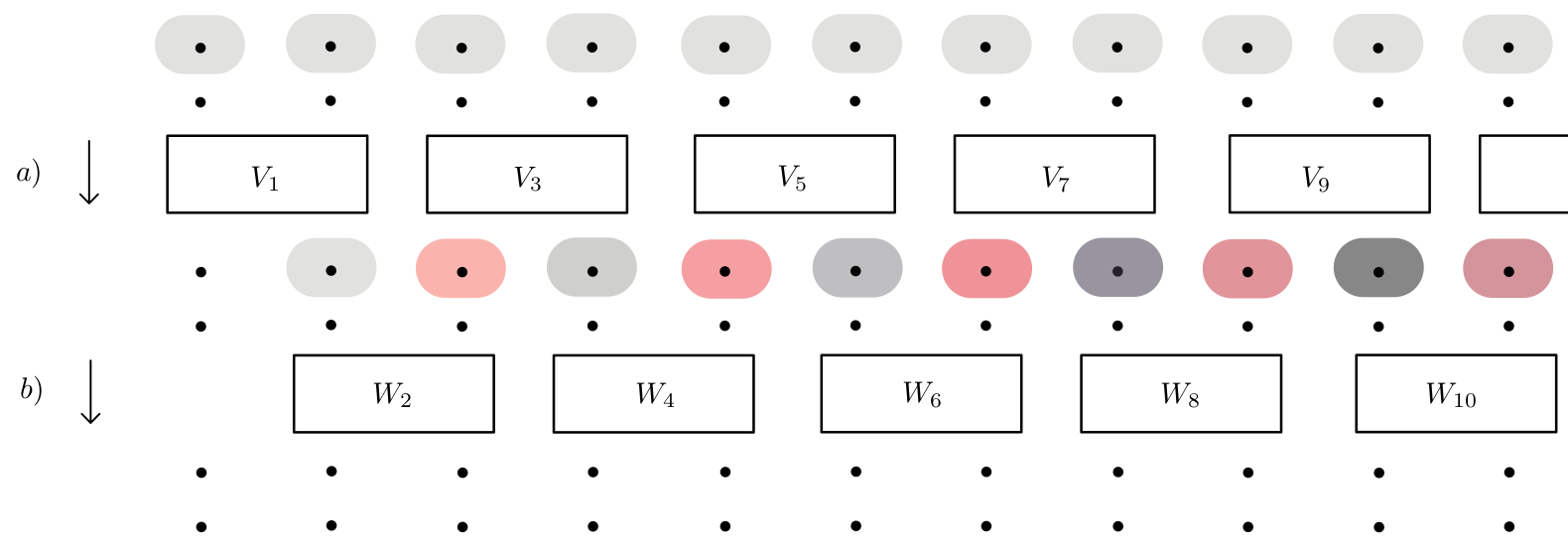}
            \caption{a) Charges are mapped into pairs charge/anticharge by unitaries $V_{2k+1}$; b) Pairs charge/anticharge are annihilated together in a $G$-invariant way into a special $G$-invariant product state via unitaries $W_{2k}$.}
            \label{fig.productstate3}
        \end{figure}

        The same arguments can be applied to the left half-chain, hence we can conclude $\phi$ is $G$-stably equivalent to a special $G$-product state. 
        
	\end{proof}

\subsection{Composition of LGA's}
We collect a few results on LGA's that will prove useful. 
As already remarked in Lemma \ref{lem.lgaproperties}, families of LGA's can be inverted and composed. 
Of course, a composition of infinitely many LGA's is in general not well-defined as an automorphism, let alone an LGA. Let us now focus on LGA's generated by a TDI that is anchored, as defined in section \ref{ALEs}. The following lemma states that in this case, an infinite composition is well-defined if the LGA's are anchored in a sufficiently sparse set:

\begin{lemma} \label{lem.infiniteproduct}
        Let $\{\alpha_{G_j,s}\}_{j=0,1,2,\ldots},\ s\in [0,1],$ be a family of LGA's  on $\AA$, such that the TDI's $G_{j}$ generating $\alpha_{G_j}$, satisfy 
        $$\sup_j \norm{G_{j}}_{\{j\},f}<\infty, $$
        for some $f \in \mathcal F$.
        Then there is $\ell_0\in \NN$ such that, if $\ell\geq \ell_0$, the limit 
		\begin{equation}\label{eq: product of m lgas}
			\lim_{M\to \infty} \alpha_{G_0,s}  \circ \alpha_{G_\ell,s} \circ \dots \circ\alpha_{G_{M\ell},s}  (A), \qquad A \in \AA
        \end{equation}
      exists and can be written as $\alpha_{G,s}(A)$ where $\alpha_{G,s}$ is a family of LGA's generated by a TDI $G$ satisfying $\norm{G}_{\tilde f} < \infty$, with $\tilde f$ depending only on $f$.
\end{lemma}
This lemma is proven in \cite{sopenkoindex} and used (implicitly) in \cite{kapustin2021classification}. It is interesting to note, for example, that the restriction to large $\ell$ is crucial.  Indeed, let $\delta_n$ be the LGA that swaps operators at sites $n-1,n$ and acts trivially on other sites. I.e.\ $\alpha_{n}(A_{n-1}\otimes B_{n})= B_{n-1}\otimes A_{n}$ where $A_{j},B_{j} \in \AA_j$ are copies of a one-site observables $A,B$.  
Then, 
$$
   \delta_{0}  \circ  \dots  \circ \delta_{M-1} \circ  \delta_{M}  (A_{M}) =A_0
$$
and therefore \eqref{eq: product of m lgas} cannot converge to an LGA when $\alpha_j=\delta_j$ and $\ell=1$.  However, if we set $\alpha_j=\delta_j$ and $\ell=2$, then \eqref{eq: product of m lgas} does converge to an LGA. Indeed, the limit 
$$
 \delta_{0}  \circ  \dots  \circ  \delta_{2M-2} \circ  \delta_{2M} 
$$
is a trivial LGA as described in Section \ref{sec.trivialTDI}.

To set the stage for the proof of the lemma, we first note by direct computation (see Lemma \ref{lem.lgaproperties}) that $\alpha_{G_0,s}  \circ \alpha_{G_\ell,s} \dots \circ\alpha_{M\ell,s}$ is an LGA generated by the operator
$$
Z_{M\ell}(s) + (\alpha_{G_{M\ell},s})^{-1} (Z_{(M-1)\ell }(s)) + \dots + (\alpha_{G_{M\ell},s} \circ \dots \circ \alpha_{G_{\ell},s}  )^{-1} (Z_0(s)).
$$
where $Z_j(s)$ is the time-dependent operator that generates the  LGA  $\alpha_{G_j,s}$, i.e.\  $Z_j(s)=\sum_{S} G_j(s,S)$, see section \ref{ALEs}. 
One  needs hence to prove that this operator corresponds to a TDI that is bounded in some $\tilde f$-norm, uniformly in $M$. 
Therefore, Lemma \ref{lem.infiniteproduct} follows from Lemma \ref{lem.infiniteproduct simple} below, which was proved in \cite{sopenkoindex}. \par 
Given a region $X \subseteq \ZZ$, we denote by $X_r$ the $r$-fattening of $X$, defined by $X_r := \{j \in \ZZ\ |\ \text{dist}(j,X) \le r\}$. In the following, $\Pi_{X_n}$ is the conditional expectation map, introduced in Proposition \ref{prop.conditionalexpectation}.

\begin{lemma} \label{lem.infiniteproduct simple}
      There is $\ell_0\in \NN$ such that, if $\ell\geq \ell_0$,  the limit 
		\begin{equation}
			\alpha(A) := \lim_{M\to \infty}  \alpha_{G_0,s}  \circ \alpha_{G_\ell,s} \circ \dots \circ\alpha_{G_{M\ell},s}(A), \qquad A \in \AA_X,\ X \in \mathcal P_{\text{fin}} (\ZZ),
        \end{equation}
      exists, and satisfies
      \begin{align*}
            \norm{\alpha (A) - \Pi_{X_n} \alpha (A) } \le \norm{A} \tilde f(n),
    \end{align*}
    for $\tilde f$ depending only on $f$.
\end{lemma}


\subsection{Locality of states}
    Short-range entangled states have a natural $f$-induced decay of correlations between distant observables, as stated below, since they are adiabatically connected to a zero correlation product state. 

    \begin{proposition} \label{lem.decayofcorrelations}
        Let $(\psi,\AA)$ be $f$-SRE, for $f \in \mathcal F$. Then 
		\begin{equation} \label{eq.decayofcorrelations}
			|\psi(AB) - \psi(A)\psi(B)| < \norm{A}\norm{B} \tilde f(\text{dist}(X,Y)), \qquad A \in \AA_X,\ B \in \AA_Y,
		\end{equation}
        for $\tilde f \in \caF$ depending only on $f$, and $X,Y \subset \ZZ$.
    \end{proposition}
    \begin{proof}
        The proof is a straightforward consequence of the Lieb-Robinson bounds \cite{liebrobinson,nachtergaele2010liebrobinson}. Denote by $H$ the TDI generating $\psi$ from a product state $\phi$, and abbreviate $\alpha_H = \alpha$.  Without loss, we can assume that $\psi(A)=\psi(B)=0$. Then
		\begin{align*}
			|\psi(AB)| = |\phi \circ \alpha (AB)| &\le | \phi\left[{\alpha(AB) - \Pi_{X_n} \left[\alpha(A)\right] \Pi_{Y_n}\left[\alpha(B)\right]}\right]| \\
			&+ |\phi\left[\Pi_{X_n} \left[\alpha(A)\right] \right] \cdot \phi \left[ \Pi_{Y_n}\left[\alpha(B)\right]\right] - \phi(\alpha(A)) \phi(\alpha(B))|,
		\end{align*}
        where $\Pi$ is the conditional expectation map of Proposition \ref{prop.conditionalexpectation}. We choose $n = \text{dist} (X,Y)/3$ so that the product of conditional expectations factorize in the product state, and the above inequality holds. The above can be upper bounded by 
		\begin{align*}
			&2\left(\norm{A} \norm{\Pi_{Y_n} \alpha(B) - \alpha(B)} + \norm{B} \norm{\Pi_{X_n} \alpha(A) - \alpha(A)} \right) \\ 
            &\le \norm{A} \norm{B} \tilde f(\text{dist} (X,Y)),
		\end{align*} 
        for some $\tilde f$ depending only on $f$, where the last inequality follows from Proposition \ref{prop.conditionalexpectation}.
  \end{proof}

\spacing

The correlation bound in Proposition \ref{lem.decayofcorrelations} is stronger than the well-known exponential decay for ground states induced by the spectral gap of a Hamiltonian \cite{Hastings_2006,nachtergaelesims.exponentialdecay}, since the bound does not depend on the size of the support of the observables $A$ and $B$.  We need the following extension of Proposition \ref{lem.decayofcorrelations}: 
\begin{corollary} \label{cor.decayofcorrelationsalgebra}
    Let $(\psi,\AA)$ be $f$-SRE, for $f \in \mathcal F$, with GNS triple $(\HH,\pi,\Omega)$. Then
     \begin{equation}
            \vert  \langle \Omega, B\pi(A) \Omega \rangle - \langle \Omega, B \Omega \rangle \langle \Omega, \pi(A) \Omega \rangle \vert < \norm{A} \norm {B} \tilde f(\text{dist}(X,Y)),
        \end{equation}
    for every $A \in \AA_{X}$, $B$ in the strong closure of $\pi(\AA_Y)$, and $\tilde f \in \mathcal F$ depending only on $f$.
\end{corollary}
\begin{proof}
Without loss of generality, we assume that $\psi(A)=\langle \Omega, \pi(A)\Omega \rangle=0$. By assumption, $B$ is a strong limit of a  sequence $\{\pi(B_n)\}_n$, for $B_n \in \AA_Y$, which, by the Kaplansky density theorem, can be chosen to satisfy $\sup_n \norm{B_n} = \sup_n \norm{\pi(B_n)}\le \norm {B}$. The bound \eqref{eq.decayofcorrelations} yields then, uniformly in $n \in \NN$:
\begin{equation}
            \vert  \langle \Omega, \pi(B_n)\pi(A) \Omega \rangle \vert \leq  \norm{A} \norm {B} \tilde f(\text{dist}(X,Y)).
        \end{equation}
and hence the bound is inherited by the $n\to \infty$ limit of the left-hand side, which is $ \vert  \langle \Omega, B\pi(A) \Omega \rangle \vert $.
\end{proof}

\spacing 

It will be crucial for us to connect SRE states via anchored TDI's, whenever this is possible. The next definition provides a condition that enables this. 
\begin{definition} \label{def.fcloseness}
    (Locality of states)
    Given $f\in \mathcal F$, two pure states $(\psi,\AA)$ and $(\phi, \AA)$ are said to be \textit{$f$-close} far from region $X\subset \ZZ$ if 
    \begin{equation}
        |\psi (A) - \phi(A)| < \norm{A} f(r) 
    \end{equation}
    holds for every $r>0$ and $A \in \AA_{(X_r)^c}$. 

    \end{definition}
We will exploit this definition in subsection \ref{sec: mutually normal g states}.


\subsection{Mutually Normal G-States}\label{sec: mutually normal g states}

A state $(\psi,\AA)$ is said to be normal with respect to $(\psi',\AA)$ whenever there is density matrix $\rho_\psi \in \mathcal B(\HH_{\psi'})$, with $\HH_{\psi'}$ the GNS representation space of $\psi'$, such that
$$\psi(A) = \Tr (\rho_\psi \pi_{\psi'} (A)), \qquad A \in \AA.$$
Moreover, if both states are pure, then there exists a $\xi_\psi \in \mathcal H_{\psi'}$, for which $\rho_\psi = \ket{\xi_\psi} \bra{\xi_\psi}$. Hence, normal pure states can be compared as vectors in a Hilbert space. The next proposition was proved in \cite{bratteliI}, see also \cite{kapustin-sopenko.flux,wojciech.thoulesspump}.
\begin{proposition}
    \label{remark.fcloseness.impliesnormality}
    Let $(\psi,\AA)$ and $(\psi', \AA)$ be two pure SRE states. Assume they are $f$-close far from a finite region $X \subset \ZZ$, for some $f \in \mathcal F$. Then $\psi'$ is normal with respect to $\psi$.
\end{proposition} 

Let us now include symmetry, by considering two mutually normal $G$-states $(\psi,\AA,U)$ and $(\psi',\AA,U)$. By Lemma \ref{lem.strongcontinuity.unitaryimplementation}, the $G$-action is implemented by a strongly continuous unitary representation $u(G) $ on $\mathcal H_{\psi'}$, such that the  vector representant $\Omega_{\psi'}$ of $\psi'$ is invariant: $u(g) \Omega_{\psi'} = \Omega_{\psi'}$ for $ g \in G$. The state $\psi$, by assumption of pureness and normality, can be implemented as
$$\psi(A) = \langle \xi_\psi, \pi_{\psi'} (A) \xi_\psi \rangle, \qquad A \in \AA,$$
for some $\xi_\psi \in \mathcal H_{\psi'}$. Furthermore, its $G$-invariance 
\begin{align*}
    \psi\circ \beta_g (A) = \langle  \xi_\psi, u^*(g)  \pi_{\psi'}(A) u(g) \xi_\psi \rangle = \psi (A), 
\end{align*}
 implies that $\xi_\psi$ transforms one-dimensionally as $u(g)\xi_\psi = q[\psi/\psi'](g) \xi_\psi$. The continuous one-dimensional representation $q[\psi/\psi'] \in \hom(G,U(1))$ defines the \textit{relative 0-dimensional charge of $\psi$ with respect to $\psi'$}. It is straightforward to establish the following properties of the $0$-dimensional charge (see \cite{wojciech.thoulesspump}). 
 
\begin{proposition} \label{prop.0dimcharge}
    For each ordered pair of mutually normal $G$-states $(\psi,\AA,U)$ and $(\psi',\AA,U)$, the continuous $q\left[ \psi/\psi' \right] \in \hom(G,U(1))$, satisfies
    \begin{enumerate} [label=\roman*)]
        \item $q\left[ \psi/\psi \right] = 1$,
        \item if $(\tilde \psi,\tilde \AA,\tilde U)$ and $(\tilde {\psi'},\tilde{\AA},\tilde {U})$ are also mutually normal, it holds that
        $$q[\psi \stack \tilde \psi/\psi' \stack \tilde{\psi'}] = q[\psi/\psi'] \cdot q[\tilde \psi/\tilde {\psi'}].$$
    \end{enumerate}
\end{proposition}

\spacing 

$0$-dimensional charges are not relevant for the $G$-equivalence relation because they are not preserved by stacking with $G$-product states, as we point out now:

\begin{lemma} \label{lem.moddingout0charge}
    Let $(\psi, \AA, U)$ and $(\psi', \AA, U)$ be an ordered pair of mutually normal G-states, as in Proposition \ref{prop.0dimcharge}. For any site $j \in \ZZ$, there are $G$-product states $(\phi^{(c)}_\psi,\AA^{(c)},U^{(c)})$ and $(\phi^{(c)}_{\psi'},\AA^{(c)},U^{(c)})$, differing only on site $j$, such that 
        $$q[\psi \stack \phi^{(c)}_\psi / \psi' \stack \phi^{(c)}_{\psi'}] = 1.$$
\end{lemma}
\begin{proof}
    Fix $j \in \ZZ$ and define an on-site representation
    $$\HH_j\charge := \CC e_j \oplus \CC w_j,$$
    with
    \begin{align*}
        U_j\charge := \begin{pmatrix}
            1 & 0 \\ 0 & \overline{q\left[ \psi / {\psi'} \right]}
        \end{pmatrix}.
    \end{align*}
    For $i \neq j$, let $\HH_i\charge := \CC e_i$ be a trivial representation. We construct the chain algebra $\AA\charge$ with on-site Hilbert spaces $\HH_k\charge,\ k \in \ZZ$, and we define the $G$-product states
    \begin{align*}
        &\phi_\psi \charge (A_k) := \langle e_k, A_k e_k \rangle,\qquad A_k \in \AA_{\{k\}},
    \end{align*}
    and
    \begin{align*}
        &\phi_{\psi'}\charge (A_k) := \begin{cases}
            \langle w_j, A_k w_j \rangle, &\qquad k = j, \\
            \langle e_k, A_k e_k \rangle, &\qquad \text{otherwise}.
        \end{cases}
    \end{align*}
    One checks that $q\left[ \phi_\psi\charge / \phi_{\psi'}\charge \right] = \overline{q\left[ \psi/{\psi'} \right]}$. The proof of the lemma now follows from part (ii) of Proposition \ref{prop.0dimcharge}.
\end{proof}

\spacing

Nevertheless, one can easily check that the zero-dimensional charge is an obstruction to connecting SRE states via an anchored, $G$-equivariant  TDI.  Once this obstruction is removed, one can connect mutually normal $G$-states, as we state now.
    
    \begin{proposition} \label{prop.connectingclosestates} 
        Let $(\psi,\AA,U)$ and $(\psi', \AA,U)$ be two f'-SRE $G$-states that are $f$-close far from site $j\in \mathbb Z$. Assume further that $q\left[ \psi/\psi' \right] = 1$. Then there exists a $G$-equivariant TDI $H$ satisfying
            $$\norm{H}_{\{j\}, \tilde f} < 1,$$ 
        such that
        \begin{equation}
            \psi = \psi' \circ \alpha_H,
        \end{equation}
        for some $\tilde f$ depending only on $f,f'$. Hence the $G$-states are $G$-equivalent. In case $q[\psi/\psi'] \neq 1$, the same holds if we replace $G$-equivalence by $G$-stable equivalence, as follows from Lemma \ref{lem.moddingout0charge}. 
    \end{proposition}
    \begin{proof}
        When $\psi$ is a product state, this was proved in \cite{wojciech.thoulesspump}. Now we extend the result to SRE states. Let $\psi = \phi \circ \alpha_{H'}$ be an $f'$-SRE G-state, with $\phi$ a special G-product state and $H'$ a G-equivariant TDI satisfying $\norm{H'}_{f'} < \infty$. Firstly, it follows from the assumptions and from Proposition \ref{prop.conditionalexpectation} that $\psi' \circ (\alpha_{H'})^{-1}$ is $\tilde f$-close to $\phi$ far from $j \in \ZZ$, for some $\tilde f \in \mathcal F$ depending only on $f,f'$.
        
        Since $\phi$ is product, there is a G-equivariant TDI $H$, that is anchored at site $j$, satisfying $\norm{H}_{\tilde f}<\infty$, for some updated $\tilde f$ depending only on $f,f'$, such that $\psi' \circ (\alpha_{H'})^{-1} \circ \alpha_H = \phi$. It follows that
        \begin{align*}
            \psi' &= \phi \circ (\alpha_H)^{-1} \circ \alpha_{H'} \\
            &= \psi \circ (\alpha_{H'})^{-1} \circ (\alpha_H)^{-1} \circ \alpha_{H'}.
        \end{align*}
        The G-equivariant LGA $(\alpha_{H'})^{-1} \circ (\alpha_H)^{-1} \circ \alpha_{H'}$ is generated by a G-equivariant TDI with finite $\tilde f$-norm, for some updated $\tilde f$ depending only on $f,f'$, and that is anchored at site $j$ (see, for instance, Lemma 2.2 of \cite{sopenkoindex}). 
    \end{proof}

\spacing

For SRE states, $f$-closeness in half-chains implies $\tilde f$-closeness on the entire chain, as we state now. This will make it easier to apply the above proposition in practice.
\begin{lemma} \label{lem.clonesshalfchains}
    Let $(\psi,\AA)$ and $(\psi',\AA)$ be $f'$-SRE states satisfying
    \begin{align*}
        |\psi(A) - \psi'(A)| < f(r) \norm{A}, \qquad A \in \mathcal A_{(-\infty,-r]} \cup \AA_{[r,\infty)},
    \end{align*}
    for some $f',f \in \mathcal F$. Then $\psi$ is $\tilde f$-close to $\psi'$ far from $\{0\}$, for $\tilde f$ depending only on $f$ and $f'$.
\end{lemma}

\begin{proof}
    It can be checked, with the help of Theorem \ref{thm.lgadecomposition}, that an $f'$-SRE state $(\psi,\AA)$ satisfies
    \begin{align} \label{eq.fclosenessglobal.1}
        \norm{\psi|_{[-r,r]^c} - \psi|_{(-\infty, -r)} \otimes \psi|_{(r,\infty)} } < f'(r),\qquad r \in \ZZ^+,
    \end{align}
    (recall the metric \eqref{eq: metric on states} on states), where $\psi|_X$ denotes the restriction of $\psi$ to the algebra $\AA_X$, $X \subset \ZZ$. Now, take $(\psi,\AA)$, $(\psi',\AA)$ as stated in the lemma. We must bound: 
    \begin{align*}
        \norm {\psi|_{[-r,r]^c} - \psi'|_{[-r,r]^c}} &\le 2f'(r) + \norm{\psi|_{(-\infty,-r)} \otimes \psi|_{(r,\infty)} - \psi'|_{(-\infty,-r)} \otimes \psi'|_{(r,\infty)}} \\
        &\le 2f'(r) + \norm{\psi|_{(-\infty,-r)} \otimes \left( \psi|_{(r,\infty)} - \psi'|_{(r,\infty)} \right)}+ \\  &+ \norm{\left( \psi|_{(-\infty,-r)} - \psi'|_{(-\infty,-r)} \right) \otimes \psi'|_{(r,\infty)}} \\
        &\le 2f'(r) + \norm{\psi|_{(r,\infty)} - \psi'|_{(r,\infty)}} + \norm{ \psi|_{(-\infty,-r)} - \psi'|_{(-\infty,-r)} }
    \end{align*}
    From the assumption, the above right hand side is then upper bounded by $\tilde f(r) = 2f'(r) + 2f(r)$. 
    
\end{proof}

\subsection{Schmidt decompositions} \label{sec.schmidtdecomposition}



An SRE state $(\psi,\AA)$ has a split GNS representation \eqref{eq.splitgns} with respect to the bipartition $\mathbb{Z}=L\cup R$, which yields a tensor product structure $\HH_L \otimes \HH_R$ for the GNS Hilbert space. We consider the corresponding Hilbert-Schmidt decomposition
		\begin{equation} \label{eq.schmidtdecomposition}
			 {\Omega_\psi} := \sum\limits_{n=1}^\infty \sqrt{\lambda_n} {\Omega^{L}_n} \otimes  {\Omega^{R}_n}
		\end{equation}
		of the GNS representant $\Omega_\psi$ of $\psi$. Here, $\Omega_n^{L/R}$ are orthonormal bases and $\lambda_n \ge 0$, $\sum_{n=1}^\infty\lambda_n=1$. The Schmidt coefficients $\lambda_n$ are ordered such that $\lambda_1 \ge \lambda_2 \ge \dots$

	\begin{lemma} \label{lem.lowerboundschmidt}
		Let $(\psi,\AA)$ be an $f$-SRE state. There is a lower bound $1/k >0$ for the first Schmidt coefficient $\lambda_1$ of ${\Omega_\psi}$, which depends only on the function $f$.
	\end{lemma}
	\begin{proof}
 In \cite{kapustin2021classification}, it was proven that  \begin{equation}
			\sum\limits_{n>N} \lambda_n < \tilde f(\log N).
		\end{equation}
  for $\tilde f$ depending only on $f$. 
 The above bound and normalization of the reduced density matrix  $\sum\limits_{j=1}^\infty \lambda_j = 1$
		 imply that $\lambda_1 \ge {\dfrac {1-\tilde f (\log N)}N}$. Choose $N$ large enough so that $1-\tilde f(\log N) > 0$, which yields a strictly positive lower bound.  
	\end{proof}

    \spacing

    Equation \eqref{eq.schmidtdecomposition} implies that the restriction $(\psi|_R, \AA_R)$ of $(\psi,\AA)$ to the halfline $R$ can be represented by a density matrix
    \begin{align*}
        \rho_R = \sum\limits_{n=1}^\infty \lambda_n \ketbra{\Omega_n^R} \in \mathcal B(\HH_R),
    \end{align*}
    via $\psi|_R (A) = \Tr (\rho_R \pi^R_\psi(A))$, $A \in \AA_R$. In spectral notation,
    \begin{equation} \label{eq.spectraldecompositionofrestriction}
        \rho_R = \sum\limits_{m=1}^\infty \lambda_{n_m} P_m,
    \end{equation}
    where $\{\lambda_{n_m}\}_m$ is a subsequence obtained from $\{\lambda_n\}_n$ by removing degeneracies, i.e.\ 
    $n_1=1$ and $n_m=\min \{n |\, n > n_{m-1}, \lambda_n < \lambda_{n_{m-1}} \}$ for $m>1$.
  The spectral projectors are 
    $$P_{m} = \sum\limits_{n_m \leq n <n_{m+1}} \ketbra{\Omega_n^R}.$$

    Moreover, recall that Lemma \ref{lem.lowerboundschmidt} implies there is a (uniform in $j$) lower bound for the absolute value of $\lambda_1$. By the normalization of the density matrix $ \Tr(\rho_R)=1$, this  immediately implies 
	\begin{equation} \label{eq.lowerboundfirstschmidt}
	\text{dim}(\text{Range}(P_1)) \le 1/\lambda_1 \le k.
	\end{equation}


\section{Proof of Completeness of Classification} \label{sec.proofmain}

In this section, we prove Theorem \ref{thm.main}. 
    \subsection{Approximation of G-states by Schmidt eigenvectors}

    The next proposition introduces a novel approach for the completeness proof when compared to \cite{kapustin2021classification}.  The main idea is to approximate a $G$-state $(\psi,\AA,U)$ with trivial index by its Schmidt eigenvectors corresponding to the largest Schmidt eigenvalue. 
    We use the split GNS representation introduced in subsection \ref{sec.constructingcupreps} and used in subsection \ref{sec.schmidtdecomposition}. In particular, we have the representation $(\mathcal H_R,\pi^R_\psi)$ of $\mathcal A_R$ and the density matrix $\rho_R \in \mathcal B(\mathcal H_R)$ representing the restriction $\psi|_R$ to $\mathcal A_R$.  Finally, we have the unitary action $u^R$ on $\mathcal H_R$ introduced in Lemma \ref{lem.splitunitary} and   the spectral decomposition \eqref{eq.spectraldecompositionofrestriction} of the density matrix $\rho_R$.
    
    \begin{proposition} \label{thm.singlets}
        Let the $G$-state $(\psi,\AA,U)$ have a trivial index: 
        $\sigma_{(\psi,\AA,U)} = [1].$ Then
        \begin{enumerate} [label=(\roman*)]
            \item Let $\mathcal J:=\text{Ran} (P_1) \subset \mathcal H_R$ with $P_1$ as in \eqref{eq.spectraldecompositionofrestriction}.  Then $\mathrm{dim}(\mathcal J) \leq k <\infty$ with $k$ depending only on $f$.  
            Moreover, $(\mathcal J, u^{R})$ is a finite-dimensional representation of $G$. 

             \item If $\xi,\xi' \in \mathcal J$ with $||\xi||=||\xi'||=1$,  then
             \begin{equation}
                 |\langle \xi, \pi^R_\psi (A) \xi' \rangle - \langle \xi|\xi'\rangle \psi (A) | < \norm {A} \tilde f (r), \qquad A \in \AA_{[j+r,\infty)},
             \end{equation}
             for $\tilde f$ depending only on $f$. 
        \end{enumerate}
    \end{proposition}
    \begin{proof}
    (i) 
    By Remark \ref{remark.triviality}, $(\HH_R,u^R)$ is a strongly-continuous linear representation of $G$. 
	By a corollary (Corollary 9.5 of \cite{knapp2002lie}) of Peter-Weyl's theorem, the GNS Hilbert space $\HH_{R}$ splits into an orthogonal countable direct sum of irreducible unitary finite-dimensional representations
	\begin{equation}
		\HH_{R} = \bigoplus\limits_{\alpha=1}^\infty \HH_\alpha.
	\end{equation}
    Since $\rho_{R}$ is G-invariant, each $\HH_\alpha$ is contained in $\text{Range}(P_i)$ for some $i \in \mathbb N$. In particular, $\mathcal J=\text{Range}(P_1)$ carries a strongly continuous linear  representation of G. By \eqref{eq.lowerboundfirstschmidt}, we conclude $\text{dim} (\mathcal J) \le k < \infty$, for $k$ depending only on $f \in \mathcal F$. \\


    \noindent (ii) 
    Recall the Schmidt decomposition in section \ref{sec.schmidtdecomposition}:
    $$
     {\Omega_\psi} := \sum\limits_{n=1}^\infty \sqrt{\lambda_n} {\Omega^{L}_n} \otimes  {\Omega^{R}_n}.
    $$
    By the result above, $\lambda_1=\lambda_2=\ldots=\lambda_{k_0}$ with $k_0=\text{Ran}(P_1) <k$.
We can expand $\xi=\sum_{n=1}^{k_0} c_n \Omega^R_{n}$, and we associate to $\xi$ the vector $\xi_{L} \in \mathcal H_L$ given by $\xi_{L}=\sum_{n=1}^{k_0} {c_n} \Omega^L_{n}$. In the same way, we associate $\xi'_{L}$ to $\xi'$ (Alternatively, we could have adjusted the Schmidt decomposition).
We now consider  $B=|\xi'_{L}\rangle\langle {\xi_{L}}| \in \mathcal B(\HH_L)$ so  that 
\begin{equation}\label{eq: correlation applied}
  \langle  \Omega, B \otimes \pi_R(A) \Omega \rangle= \lambda_1 \langle \xi',    \pi_R(A) \xi \rangle.
\end{equation}
By the discussion in section \ref{sec.constructingcupreps}, we know $\mathcal B(\HH_L) = \pi_L(\AA_L)''$.   
We can then apply Corollary \ref{cor.decayofcorrelationsalgebra} to \eqref{eq: correlation applied} to obtain: 
    \begin{align*}
        \dfrac 1{\lambda_1} \tilde f(r) \norm{A} \norm{B} \ge 
       | \langle \xi',    \pi_R(A) \xi \rangle-   
       \langle \xi',  \xi \rangle
       \langle  \Omega, \pi(A) \Omega \rangle |.
    \end{align*}
    Due to the lower bound \eqref{eq.lowerboundfirstschmidt} for $\lambda_1$ and $||B||\leq 1$, the claim follows upon updating $\tilde f$ depending only on $f$. 
\end{proof}


Let us recall the family of bipartitions $\bbZ=L^{(j)}\cup R^{(j)}$.  Since the translates of an $f$-SRE state are still $f$-SRE states, and also the property of being a $G$-state is preserved, we can apply the above theorem to any of these bipartitions and obtain Hilbert spaces $\mathcal J^{(j)}$ with $\mathrm{dim}(\mathcal J^{(j)})\leq k$, uniformly in $j$. 
For the next result, we are in the same position as for Proposition \ref{thm.singlets}, i.e.\ we choose a $G$-state $(\psi,\caA,U)$  that is $f$-SRE, and whose index is trivial. Without loss of generality, via Lemma \ref{lem.specialgstate}, we also recall that there exists a special $G$-product state  $(\phi,\caA,U)$.
\begin{proposition} \label{cor.invariantstates}
		Let the assumptions of Proposition \ref{thm.singlets} be satisfied.  There is a $\tilde{f}$ depending only on $f$, 
  a $G$-product state $(\phi',\caA',U')$ and a family  of $G$-states
   $(\gamma_j, \AA \stack \AA',U\stack U')$, indexed by $j\in \ZZ$,  satisfying:
		\begin{enumerate}
  \item Each $\gamma_j$ is $\tilde{f}$-SRE.
  \item $\gamma_j$ is factorized between $L^{(j)}$ and $R^{(j)}$.
			\item $\gamma_j$ is equal to the $G$-product state $\phi \tilde \otimes \phi'$ on $L^{(j)}$:
			\begin{equation}
				\phi \tilde \otimes \phi'(A) = \gamma_j (A), \qquad A \in \AA\tilde \otimes \AA'|_{(-\infty, j-1]};
			\end{equation}
			\item $\gamma_j$ is $\tilde f$-close far from $j$ to $\psi \tilde \otimes \phi'$, to the right of $j$: 
			\begin{equation}
				|\psi \tilde \otimes \phi'(A) - \gamma_j (A)| < \tilde f(r) \norm{A}, \qquad A \in \AA\tilde \otimes \AA'|_{[j+r,\infty)},
			\end{equation}
		\end{enumerate}
	\end{proposition}
	
	\begin{proof}
 We first define $\AA'$ and $U'$. 
    The  on-site Hilbert space $\HH'_{i}$ and the on site-action $U_i'$ are such that 
    $$\HH'_{i} := \overline {\mathcal J^{(i)}} \oplus \CC v_i,$$
    where $(\overline {\mathcal J^{(i)}}, U_i')$ is the conjugate representation of ${\mathcal J^{(i)}}$ (the representation obtained in Proposition \ref{thm.singlets} w.r.t a bipartition $L^{(i)}$, $R^{(i)}$) and $(\CC v_i, U_i')$ is an arbitrary continuous one-dimensional representation  of $G$. 
 The chain algebra $\AA'$ is then constructed from the on-site Hilbert spaces $\HH'_i$.
 The product state $\phi'$ is defined by the formal product 
 $$
 \phi'=\otimes_{j\in \ZZ} \langle v_j,(\cdot) v_j\rangle,
 $$
as described in section \ref{sec: states}. \par

Fix $j \in \ZZ^+$ and denote $R^{(j)} = R,L^{(j)} = L $ to avoid clutter. Consider a GNS triple $(\HH'_R, \pi^R_{\phi'},\Omega_{\phi'})$ of $\phi_R'$, and the subspace
$$\mathcal G_j := \pi_{\phi'}^R (\AA_{\{j\}}) \Omega_{\phi'} \subset \HH_R'.$$
Since $\phi'$ is a $G$-product state, we can define a Hilbert space isomorphism 
$$
p: \HH_j' \to \mathcal G_j, \qquad p(w)=  \pi_{\phi'}^R (| w\rangle \langle v_j |) \Omega_{\phi'}, \qquad w \in \HH_j'. $$
This $p$ interwtines the unitary $G$-actions. 
We denote $p(\overline{\mathcal J^{(j)}}) \subset \HH_R'$ simply by $\overline{\mathcal J^{(j)}}$ since confusing is unlikely. 
We then consider the subspace $$\mathcal J^{(j)} \stack \overline{\mathcal J^{(j)}} \subset \HH_R \stack \HH_R'.$$
Let ${\eta^j} \in \HH_{R} \tilde \otimes \HH_{R}'$ be the singlet representation of $G$ arising in $ \mathcal J^{(j)} \tilde \otimes \overline {\mathcal J^{(j)}}$ (see  \cite{fuchs,creutz} for its existence). We are finally ready to define the factorized state $\gamma_j=\gamma_j|_{L}\otimes \gamma_j|_{R}$.
On the left, we simply set $\gamma_j|_{L}=(\phi \tilde \otimes \phi')_{L}$ (which is indeed pure since $\phi,\phi'$ are product states), and on the right we set 
\begin{equation}
    \gamma_j|_{R}=  \left< \eta^j, \pi_\psi^R \tilde\otimes \pi_{\phi'}^R (\,\cdot\,) \eta^j\right>.
\end{equation}
In this way, items 2. and 3. are satisfied. 
To check item 4., we decompose $\eta^j=\sum_{\ell=1}^{k_0} \xi_\ell \stack \alpha_\ell$ by the Schmidt decomposition for the tensor product $\mathcal J^{(j)} \stack \overline{\mathcal J^{(j)}}$.
Then 
$$
\gamma_j|_{R}=\sum_{\ell,\ell'=1}^{k_0} \theta_{\ell,\ell'},
$$
where 
$$
\theta_{\ell,\ell'}= \left< \xi_\ell\stack \alpha_\ell, \pi_\psi^R \tilde\otimes \pi_{\phi'}^R (\,\cdot\,) \xi_{\ell'} \stack \alpha_{\ell'}   \right>  =  
\left< \xi_\ell, \pi_\psi^R  (\,\cdot\,) \xi_{\ell'} \right> \stack 
\left<  \alpha_\ell,  \pi_{\phi'}^R (\,\cdot\,) \alpha_{\ell'}   \right>.
$$
By the fact that $\phi'$ is a product state and the orthogonality of the collection $(\alpha_{\ell})_{\ell=1,\ldots,k_0}$, we have, for $r>1$,
$$
\left<  \alpha_\ell,  \pi_{\phi'}^R (\,\cdot\,) \alpha_{\ell'}   \right>|_{[j+r,\infty)} =\delta_{\ell,\ell'}  ||\alpha_\ell||^2 \phi'|_{[j+r,\infty)}.
$$
By Proposition \ref{thm.singlets} and orthogonality of the collection $(\xi_{\ell})_{\ell=1,\ldots,k_0}$, we have
$$
||\left(\left<  \xi_\ell,  \pi_{\psi}^R (\,\cdot\,) \xi_{\ell'}   \right> - \delta_{\ell,\ell'} ||\xi_\ell||^2 \psi\right)|_{[j+r,\infty)}  || \leq  \tilde f(r).
$$
Combining these two bounds
and summing over  $\ell,\ell'$, we get 
$$
|| (\gamma_j - 
  \psi\otimes\phi')|_{[j+r,\infty)}|| 
  \leq k^2    \tilde f(r),
$$
 which yields item 4. upon updating $\tilde f$, since $k$ depends only on $f$.

    \par 
    It is left for us to prove item 1), i.e.\ that $\gamma_j$ is SRE. Let $H$ be a TDI generating $\psi \stack \phi'$ from a product state $\phi\stack \phi'$, and consider a decomposition as in Theorem \ref{thm.lgadecomposition} for $\alpha_H(1)$. The SRE state
    $$\psi \stack \phi' \circ (\alpha_{H_{L^{(j)}}})$$
    is $\tilde f$-close to $\gamma_j$ far from site $j$, for some $\tilde f \in \mathcal F$ that depends only on $f$. By Proposition \ref{prop.connectingclosestates},
    $$\gamma_j = \psi\stack \phi' \circ \alpha_{H_{L^{(j)}}} \circ \alpha_K,$$
    for some TDI $K$ satisfying
    $$\norm{K}_{j,\tilde f}<\infty,$$
    for $\tilde f$ depending only on $f$. Hence $\gamma_j$ is $\tilde f$-SRE, for some updated $\tilde f$ depending only on $f$.  
	\end{proof}

\spacing

Proposition \ref{cor.invariantstates} states that there exists a factorized state that is close to  $\psi \stack \phi'$  to the right and that equals a product state to the left. This already suggests that $\psi$ is in the trivial phase (cf.\ \cite{Gaiotto}) and this will be developed in the following sections. By similar reasoning, we can argue that $\psi \stack \phi'$  is $G$-equivalent to a factorized state that is close to $\psi \stack \phi'$ both to the left and to the right.  


\begin{corollary} \label{thm.trivialindex.equivalencetoaproduct}
    Let $(\psi,\AA,U)$ be a $G$-state with trivial $H^2$ index. Then it is $G$-stable equivalent to a $G$-state 
    $$(\gamma_0^L|_{L^{(0)}} \otimes \gamma_0^R|_{R^{(0)}}, \AA \stack \AA', U\stack U'),$$
    that is factorized between $L^{(0)}$ and $R^{(0)}$. 
\end{corollary}
\begin{proof}
    Let $\gamma_0^R$ and $\gamma_0^L$ be two G-invariant half-states obtained as in Proposition \ref{cor.invariantstates}, for the right and left chain, respectively. By Lemma \ref{lem.clonesshalfchains}, the pure state
    $$\gamma_0^L \otimes \gamma_0^R$$
    is $\tilde f$-close to $\psi \stack \phi'$. Therefore, by Proposition \ref{remark.fcloseness.impliesnormality}, it is normal with respect to $\psi \tilde\otimes \phi'$. As a result of Proposition \ref{prop.connectingclosestates}, they are $G$-stably equivalent. 
\end{proof}

\subsection{Connecting different approximants}
The G-states $\gamma_j$ constructed above manifestly have trivial index, as they are factorized between $L^{(j)}$ and $R^{(j)}$. It is fairly easy to see that $\gamma_j$ and $\gamma_{j'}$ are $G$-stably equivalent, but we need a somewhat finer property. To state it, we will choose $j'=j+\ell$ for some large $\ell$ that will be chosen later. 

 We use the same notation as of Proposition \ref{thm.singlets} and Proposition \ref{cor.invariantstates}. 

    \begin{lemma}
    
        There is a G-product state $(\phi'',\AA'',U'')$ and a family of  G-states 
        $$(\nu_{j,\ell},\AA \stack \AA'\stack \AA'',U\stack U'\stack U''),$$
        parametrized by $\ell \in \NN, j \in \ell \ZZ^+$, such that 
            \begin{enumerate} [label=(\roman*)]
            \item $\nu_{j,\ell}$ is factorized between $(-\infty,j)$, $[j,j+\ell)$, and $[j+\ell,\infty)$. 
                \item $\nu_{j,\ell}$ is $\tilde f$-close to $\gamma_{j} \stack \phi''$ far from $j+\ell$, for $\tilde {f}$ depending only on $f$.
            \end{enumerate}
    \end{lemma}
    \begin{proof}
        The states $\gamma_j,\ j \in \ell \ZZ^+$ constructed in Proposition \ref{cor.invariantstates} are, by definition, $G$-invariant $\tilde f$-SRE. Since each $\gamma_j$ is factorized between regions $L^{(j)}$ and $R^{(j)}$, they also carry trivial index. Consequently each of them satisfies the assumptions of Proposition \ref{thm.singlets}, with respect to a bipartition $L^{(j+\ell)},R^{(j+\ell)}$. Moreover, its reduced density matrix 
    $$\rho_{L^{(j+\ell)}} = \sum\limits_{m=1}^\infty \lambda_{n_m}(\gamma_j) P_m(\gamma_j)$$
    has finite rank, since it is factorized between $L^{(j)}$ and $R^{(j)}$, and it is $G$-invariant. Analogously as in Proposition \ref{cor.invariantstates}, by means of part (ii) of Proposition \ref{thm.singlets}, we argue there is an auxiliary algebra $\AA''$, a $G$-product state $\phi''$ on $\AA''$, and a $G$-invariant pure state $\zeta_{[j,j+\ell]}$ on $(\mathcal{A} \tilde \otimes \mathcal {A}' \tilde\otimes \mathcal {A}'')|_{[j,j+\ell]}$ with the property that   
    $$
    \nu_{j,\ell} := (\phi \stack \phi' \stack \phi'')|_{L^{(j)}} \otimes \zeta_{[j,j+\ell]} \otimes (\gamma_{j+\ell} \stack \phi'')|_{R^{(j+\ell)}}
    $$ 
    is $\tilde f$-close to $\gamma_j\stack \phi''|_{[j,j+\ell]}$ far from site $j+\ell$: for that, let $\mathcal J(\gamma_j) := \text{Ran }(P_1(\gamma_j))$. Then, we define
    $$
        \HH''_{k+\ell} := \CC e_{k+\ell} \oplus \overline {\mathcal J (\gamma_k)}, \qquad k \in \ell\ZZ^+, 
    $$
    as the on-site spaces defining $\AA''$, where $\CC e_{k+\ell}$ transforms one-dimensionally with $G$. \par 
    Each state $(\zeta_{[j,j+\ell]}, (\AA\stack \AA'\stack \AA'')_{[j,j+\ell]})$ is then defined analogously as in Proposition \ref{cor.invariantstates}, as a vector state corresponding to the singlet representation of $G$ arising in $\mathcal J(\gamma_j) \stack \overline{\mathcal J(\gamma_j)}$. 

    The state $\zeta_{[j,j+\ell]}$ is a pure state on a finite-dimensional matrix algebra $(\AA\stack \AA'\stack \AA'')_{[j,j+\ell]}$, hence it corresponds to a one-dimensional representation contained in 
    $$ \bigotimes\limits_{k \in [j,j+\ell]} \HH_k \stack \HH'_k \stack \HH''_k.$$
    This finishes the construction in part (i). The proof of Part (ii) follows from part (ii) of Proposition \ref{thm.singlets} and lemma \ref{lem.clonesshalfchains}.
    \end{proof}

\begin{figure} [h]
    \centering
    \includegraphics[scale=.26]{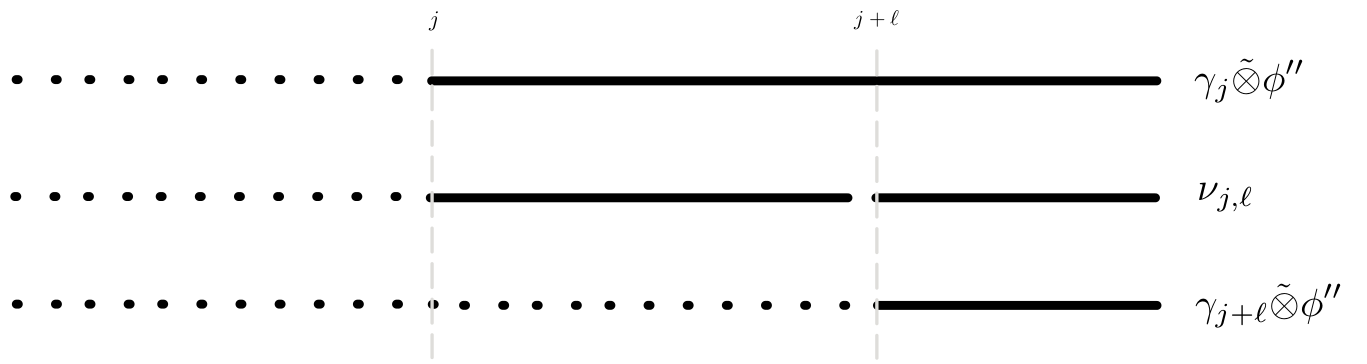}
    \caption{Locality properties of half-states. States $\gamma_j \stack \phi''$ and $\nu_{j,\ell}$ coincide on $L^{(j)}$, as a G-product state, depicted as a sequence of isolated dots. State $\nu_{j,\ell}$ is factorized between $R^{(j+\ell)}$ and $[j,j+\ell)$.}
    \label{fig.localityofstates1}
\end{figure}

\spacing
From the previous lemma, there is a sequence of equivalences
\begin{align*}
    \dots \xrightarrow{\hspace{1cm}} \gamma_{j} \stack \phi'' \xRightarrow{\hspace{1cm}} \nu_{j,\ell} \xrightarrow{\hspace{1cm}} \gamma_{j+\ell} \stack \phi'' \xRightarrow{\hspace{1cm}} \nu_{j+\ell,\ell} \xrightarrow{\hspace{1.5cm}} \dots  
\end{align*}
where a simple arrow means conjugation with a unitary on a finite box of size $\ell$, and a double arrow means the states are $\tilde f$ close far from a site depending on the states (e.g. a state $\gamma_j \stack \phi''$ is $\tilde f$-close to $\nu_{j,\ell}$ far from site $j+\ell$). Locality of such states is schematically illustrated in Figure \ref{fig.localityofstates1}, whereas their $\tilde f$-closeness is illustrated in Figure \ref{fig.closenessofstates1}. In what follows, we need to lift the equivalences to $G$-equivalences, by stacking with product states such that the zero-dimensional charge becomes constant along the above sequence. \\

\begin{figure}
    \centering
    \includegraphics[scale=.35]{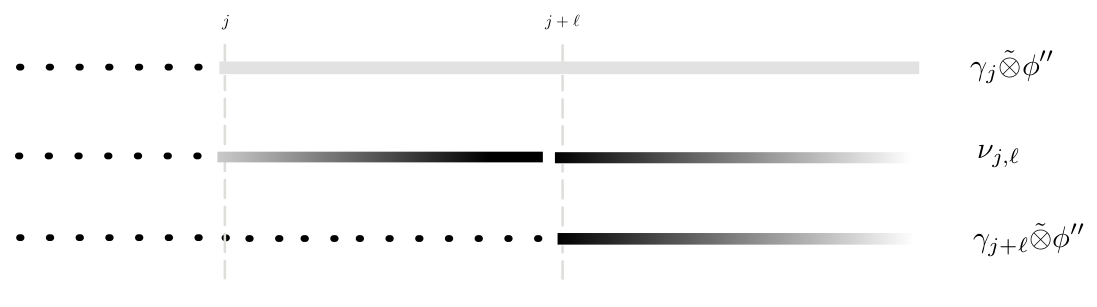}
    \caption{$\tilde f$-closeness of states. State $\nu_{j,\ell}$ is $\tilde f$-close to $\gamma_j \stack \phi''$ far from site $j+\ell$, and coincide with $\gamma_{j+\ell} \stack \phi''$ on $R^{(j+\ell)}.$}
    \label{fig.closenessofstates1}
\end{figure}

\begin{lemma} \label{lem.chargesequence}
    There are families of $G$-product states
    $$ (\phi\chargep_{j,\ell}, \AA\charge, U\charge),\ (\phi\charge_{j,\ell}, \AA\charge, U\charge),$$
    parametrized by $\ell \in \NN, {j \in \ell \ZZ^+}$, such that, defining 
 \begin{equation}
    \omega_{j,\ell} := \gamma_{j} \stack \phi'' \stack \phi_{j,\ell} \charge, \qquad    \varrho_{j,\ell} := \nu_{j,\ell} \stack \phi\chargep_{j,\ell},
\end{equation}
all states in the sequence 
\begin{align} \label{eq: sequence of omega and rho}
    \dots \xrightarrow{\hspace{1cm}} \omega_{j,\ell} \xRightarrow{\hspace{1cm}} \varrho_{j,\ell} \xrightarrow{\hspace{1cm}} \omega_{(j+\ell),\ell} \xRightarrow{\hspace{1cm}} \varrho_{(j+\ell),\ell} \xrightarrow{\hspace{1cm}} \dots 
\end{align}
are mutually normal and have zero-relative charge, i.e.\
\begin{align*}
    &q\left[ \omega_{j+\ell,\ell} / \varrho_{j,\ell} \right] = 1, \\
    &q\left[ \varrho_{j,\ell} / \omega_{j,\ell} \right] = 1.
\end{align*}
Moreover (see also  Figure \ref{fig.charges}), 
    \begin{enumerate}
        \item $\phi\chargep_{j,\ell}$ differs from $\phi\charge_{j,\ell}$ only at site $j+\ell-1$, 
        \item and $\phi\charge_{j+\ell,\ell}$ differs from $\phi\chargep_{j,\ell}$ only at site $j+\ell-1$, 
    \end{enumerate}
\end{lemma}

\begin{figure} [h]
    \centering
    \includegraphics[scale=.35]{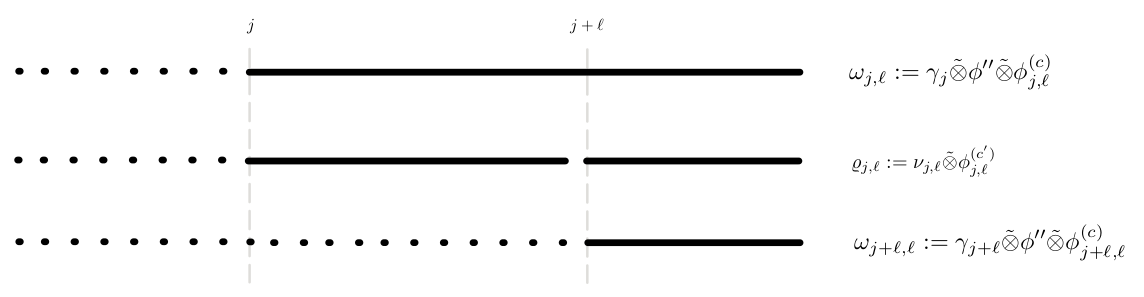}
    \caption{Upgraded version of Figure \ref{fig.localityofstates1}. States $\gamma_j \stack \phi''$ and $\nu_{j,\ell}$ are now stacked with $G$-product states, so that 0-dim charges are matched. Their locality properties remain the same.} 
    \label{fig.localityofstates2}
\end{figure}




\begin{proof} [Proof of Lemma \ref{lem.chargesequence}]
      
    We construct the auxiliary chain algebra  $\AA^{(c)}$ (with $(c)$ standing for ``charge''), as follows: The on-site spaces $  \HH_{j}^{(c)}$ and on-site unitary actions $U_{j}^{(c)}$ are 
    
    \begin{enumerate} [label=(\roman*)]
    \item  For $j \in (\ell \ZZ^+ - 1)$,
    $$
    \HH_{j}^{(c)} = \mathbb C e_{j} \oplus \mathbb C w_{j} \oplus \mathbb C v_j,
    $$
    according to an on-site unitary representation
    $$U_{j}^{(c)}(g) := \begin{pmatrix}
        1 & 0 & 0 \\ 0 & {q_j} & 0 \\ 0 & 0 & {q'_j}
    \end{pmatrix},$$

    with 
    \begin{align*}
        &q_j = {q\left[ \gamma_{j+\ell,\ell} \stack \phi''/ \nu_{j,\ell} \right]} \cdot q\left[ \nu_{j,\ell}/\gamma_{j,\ell} \stack \phi'' \right], \\
        &q'_j = {q\left[ \gamma_{j+\ell,\ell} \stack \phi''/\nu_{j,\ell} \right]}.
    \end{align*}

    \item For $j \not\in (\ell \ZZ^+-1)$,
    $$\HH_j\charge := \CC e_j$$
    where $e_j$ transforms trivially under $G$. 
    \end{enumerate}
    For each $j \in \ell \ZZ^+$, we define $G$-product states (see also Figure \ref{fig.charges})
    \begin{align} \label{eq.chargedstate}
        \phi_{j,\ell}^{(c)} (A_k) := \begin{cases}
            \langle w_k, A_k  w_k \rangle,&\qquad k \in (\ell \ZZ^+-1) \cap [j+\ell-1,\infty) \\
            \langle e_k, A_k  e_k\rangle,&\qquad \text{otherwise.} \\
        \end{cases} 
    \end{align}
    for $A_k \in \mathcal A_{\{k\}}$, and
    \begin{align}
        \phi_{j,\ell}^{(c')} (A_k) := \begin{cases}
            \langle w_k, A_k w_k \rangle,&\qquad k \in (\ell \ZZ^+-1) \cap [j+2\ell-1,\infty) \\
            \langle v_k, A_k  v_k \rangle, &\qquad k = j+\ell - 1 \\
            \langle e_k, A_k  e_k\rangle,&\qquad \text{otherwise.} \\
        \end{cases} 
    \end{align}
    
    
    \noindent The proof of the lemma now follows from item (ii) of Proposition \ref{prop.0dimcharge}, since by design we have
    \begin{equation*}
        q\left[ \phi\chargep_{j,\ell} / \phi\charge_{j,\ell}  \right] = \overline{q\left[ \nu_{j,\ell} / \gamma_{j}\stack \phi'' \right]},
    \end{equation*}
    and 
    \begin{equation*}
        q\left[ \phi\charge_{j+\ell,\ell} /\phi\chargep_{j,\ell} \right] = \overline{q\left[ \gamma_{j+\ell} \stack \phi'' / \nu_{j,\ell} \right]}.
    \end{equation*}

\end{proof}

\begin{figure} [h]
    \centering
    \includegraphics[scale=.3]{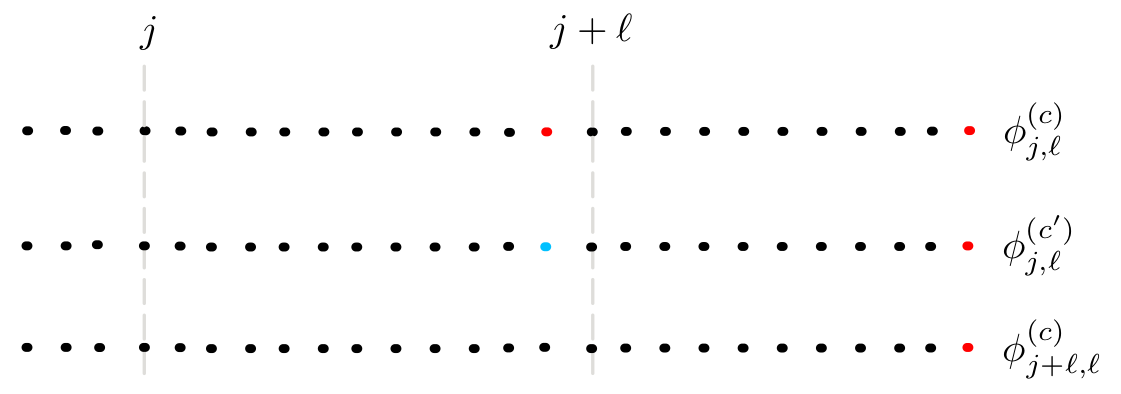}
    \caption{Definition of charged $G$-product states. Black dots represent trivial charges (on-site trivial representations), whereas colored dots represent non-trivial charges. States $\phi\charge_{j,\ell}$, $\phi\chargep_{j,\ell}$, and $\phi\charge_{j+\ell,\ell}$ differ only at site $j+\ell-1$.}
    \label{fig.charges}
\end{figure}

\begin{figure} [h]
    \centering
    \includegraphics[scale=.35]{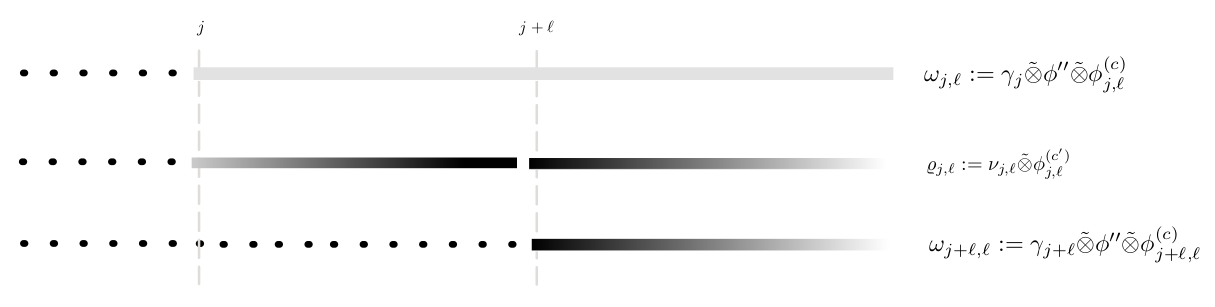}
    \caption{Upgraded version of Figure \ref{fig.closenessofstates1}. Closeness of states is preserved after stacking with charge $G$-product states. Now, they have zero relative 0-dim charge.}
    \label{fig.closenessofstates2}
\end{figure}


We will now take advantage of the sequence of states \eqref{eq: sequence of omega and rho} with equal $0$-dim relative charges, to connect these states via quasi-local operations. In particular, we show

    \begin{proposition} \label{thm.connectinghalfstates}
Let the family $\omega_{j,\ell}$ be as above, then
        \begin{align}
            \omega_{j,\ell} \circ \alpha_{G_{j}} \circ \Ad{U_{j,j+\ell}} = \omega_{(j+\ell),\ell}.
        \end{align}
        where
        \begin{enumerate} [a.]
        \item $U_{j,j+\ell}$ is a $G$-equivariant unitary  strictly localized at $[j,j+\ell]$,
          \item $G_{j}$  is a $G$-equivariant TDI satisfying
        \begin{equation*}
            \norm{G_{j}}_{\{j+\ell\}, \tilde f} < \infty,
        \end{equation*}
        for $\tilde f$ depending only on $f$, and $G_{j}(S) = 0$ whenever $S \cap L^{(j)} \neq \emptyset$.
        \end{enumerate}
	\end{proposition}







	\begin{proof} [Proof of Theorem \ref{thm.connectinghalfstates}]
    The states $\varrho_{j,\ell}$ and $\omega_{j+\ell,\ell}$ are both factorized between the regions $L^{(j)}$, $L^{(j+\ell)} \cap R^{(j)}$, and $R^{(j+\ell)}$, and they differ only on the middle of these regions, which is an interval of length $\ell$.  Since these states have zero relative charge, there is a $G$-equivariant unitary $U_{j,j+\ell} \in (\AA \stack \AA' \stack \AA'' \stack \AA\charge)_{[j,j+\ell]}$ such that 
        $$\varrho_{j,\ell} \circ \Ad{U_{j,j+\ell}} = \omega_{j+\ell,\ell}.$$
Secondly, the states $\omega_{j,\ell}$ and $\varrho_{j,\ell}$ are $\tilde f$-close far from $j+\ell$ for some $\tilde f \in \mathcal F$ and carry zero relative zero dimensional charge.  Hence 
       , they can be connected via conjugation with an $\tilde f$-anchored, $G$-equivariant LGA $\alpha_{G_{j}}(1)$ at $j+\ell$, by proposition \ref{prop.connectingclosestates}.  Moreover, since both $\varrho_{j,\ell}$ and $\omega_{j,\ell}$ factorize between $L^{(j)}$ and $R^{(j)}$, and coincide in $L^{(j)}$, the TDI can be chosen to satisfy $G_j(S)(s) = 0$ whenever $S \in \mathcal P_{\text{fin}}(L^{(j)})$.

     \end{proof}

\subsubsection{Proof of Main Theorem \ref{thm.main}}
\noindent We therefore have the ingredients to prove \ref{thm.main}. \\ 

        By Corollary \ref{thm.trivialindex.equivalencetoaproduct}, $\psi$ is $G$-stably equivalent to $\gamma_0^L \otimes \gamma_0^R$. Moreover, we use Theorem \ref{thm.connectinghalfstates} to argue that each $\gamma_0^{L(R)}$ is in the trivial phase: first of all, recall there is a sequence of $G$-equivalences
     \begin{align*}
            \omega_{0,\ell} \overset{\alpha_{G_0} \circ \Ad{U_{0,\ell}}}{\xrightarrow{\hspace{1.5cm}}} \omega_{\ell,\ell} \overset{\alpha_{G_\ell} \circ \Ad{U_{\ell,2\ell}}}{\xrightarrow{\hspace{1.5cm}}} \omega_{2\ell,\ell}  \overset{\alpha_{G_{2\ell}} \circ \Ad{U_{2\ell,3\ell}}}{\xrightarrow{\hspace{1.5cm}}} \dots
        \end{align*}
We also write 
$\omega_{\infty}=\phi \stack \phi' \stack \phi'' \stack \phi\charge_\infty$, with $\phi\charge_\infty$ being a special $G$-product state on $\AA\charge$. Then, by the definition of the states $\gamma_{j}$, it immediately follows that 
$$
\lim_{M} \omega_{M\ell,\ell}(A)=\omega_{\infty}(A), \qquad A \in \caA\stack \AA' \stack \AA'' \stack \AA\charge.
$$
We will now argue that there is a TDI $H$ such that
\begin{equation}\label{eq: application of tough lemma}
 \omega_{0,\ell} \circ \alpha_H (A)= \lim_{M} \omega_{M\ell,\ell}(A)= \omega_\infty (A)    
\end{equation}
which will conclude the proof of Theorem \ref{thm.main}.

Taking into account the supports of $U_{j,j+\ell}$, $G_j$, $j \in \ell \ZZ^+$ (cf.\ parts a. and b. of Theorem \ref{thm.connectinghalfstates}), we rewrite the automorphism relating $\omega_{N\ell,\ell}$ to $\omega_{0,\ell}$ as
\begin{align} \label{eq.infprodlocalalgebra}
            \prod\limits_{k=0}^{N} \alpha_{G_{k\ell}} \circ \Ad{U_{k\ell,(k+1)\ell}} = \left( \prod\limits_{k=0}^{N}  \alpha_{G_{k\ell}} \right) \left(  \prod\limits_{k=0}^{N}  \Ad{U_{k\ell,(k+1)\ell}} \right), 
        \end{align}
        for any $N \in \NN$. These are ordered products where the lower indices come first (e.g. as in $\prod\limits_{k=0}^{N}  \alpha_{G_{k\ell}} = \alpha_{G_0} \circ \alpha_{G_\ell} \circ \dots$).
The second factor in \eqref{eq.infprodlocalalgebra} converges strongly (i.e.\ applied to any $A\in \caA$) to an LGA $\alpha_{H}$ (cf.\ subsection \ref{sec.trivialTDI}). 
Since both factors are automorphisms (in particular, they have unit norm) and the composition of 2 LGA's is again an LGA, it suffices to prove that the first factor converges strongly to an LGA in order to get strong convergence of the product to an LGA. Hence, we need to prove that there is an TDI $G$ such that 
$$
\lim_N  \prod\limits_{k=0}^{N}  \alpha_{G_{k\ell}} (A)= \alpha_G(A)
$$
This is true provided $\ell$ is chosen large enough, as argued in Lemma \ref{lem.infiniteproduct}.

\appendix

\section{Locally Generated Automorphisms} \label{appendix.lga}

\subsection{\texorpdfstring{$\caF$}{}-functions}

Locality properties of LGAs can be understood in terms of the Lieb-Robinson bound \cite{liebrobinson}, more recently formulated by \cite{brunoamandaI, nachtergaele2010liebrobinson}. This propagation bound was usually formulated in terms of $F$-functions. Given a 1d lattice $\Gamma \subseteq \ZZ$, those are functions satisfying:
\begin{enumerate}
    \item $F: [0,\infty) \to [0,\infty)$ is non-increasing;
    \item $F$ is uniformly integrable over $\Gamma$:
    \begin{equation*}
        \sup\limits_{x \in \Gamma} \sum\limits_{y \in \Gamma} F(\text{dist}(x,y)) \le C_F' < \infty; 
    \end{equation*}
    
    \item $F$ satisfies a convolution property:
    \begin{equation*}
        \sum\limits_{z \in \Gamma} F(\text{dist}(x,z))F(\text{dist}(z,y)) \le C_F F(\text{dist}(x,y)), \qquad \text{for all } x,y \in \Gamma.
    \end{equation*}
\end{enumerate}

Denote by $\mathcal G$ the set of $F$-functions satisfying the above properties. Given an $F$-function, there is a corresponding norm $|||\cdot|||_F$ on TDI's $H$, via:
\begin{equation}
    |||{H}|||_F := \sup\limits_{t \in [0,1]} \sup\limits_{x,y\in \ZZ} \dfrac {1}{F(\text{dist}(x,y))} \sum\limits_{\substack{ S \in \mathcal P_{\text{fin}}(\ZZ)\\ S \ni \{x,y\} }} \norm{H(S)(t)}. 
\end{equation}
Recall that $\mathcal F$ is the class of non-increasing, strictly positive functions $f:\bbN^+\to\bbR^+$ that decay faster than any polynomial, as defined in section \ref{sec: TDIs}.

\begin{lemma} \label{lem.equivalenceofnorms}
    For each $f \in \mathcal F$, there is an $F$-function $F \in \mathcal F \cap \mathcal G$ such that
    \begin{equation}
        f(r) \le F(r+1), \qquad \text{and} \qquad \norm{H}_f \le |||H|||_F.
    \end{equation}
    Conversely, for each $F$-function $F \in \mathcal F \cap \mathcal G$, there is $f \in \mathcal F$ such that
    \begin{equation}
        F(r+1) \le f(r), \qquad \text{and} \qquad |||H|||_F \le \norm {H}_f. 
    \end{equation}
\end{lemma}
\begin{proof}
    The proof follows closely the proof of Lemma 7 in \cite{Bachmann2022stability}, except that the authors consider functions with streched exponential decay instead of superpolynomial decay. For that reason, given $f \in \mathcal F$, one can construct an $F$-function
    \begin{equation}
        F(r) := \min\limits_{r' \le r} \tilde F(r')
    \end{equation}  
    where
    \begin{equation}
        \tilde F(r) := \max\left[ f(r), \dfrac 1{f(1)+1}  \sum\limits_{m=1,\dots,r-1} F(r-m)F(m) \right],
    \end{equation}
    which satisfies all the necessary axioms and inherits the fast decay from $f$. The proof then follows from the reference.
\end{proof}

\spacing

\subsection{Conditional expectations}

Recall that, given a region $X \subseteq \ZZ$, the set $X_r$ is the $r$-fattening of $X$, defined by $X_r := \{j \in \ZZ\ |\ \text{dist}(j,X) \le r\}$. The next proposition and theorem are consequences of the Lieb-Robinson bound \cite{liebrobinson,nachtergaele2010liebrobinson}, which can be now applied in our setting, due to the previous Lemma \ref{lem.equivalenceofnorms}. 


\begin{proposition} \label{prop.conditionalexpectation}
    Assume $\norm{H}_f < \infty$ for some $f \in \mathcal F$. For each $X \in \mathcal P_{\text{fin}}(\ZZ)$ and $n \in \NN$, there exists a linear map 
    $$\Pi_{X_n}: \AA \to \AA_{X_n}$$
    satisfying
    \begin{equation}
        \norm{\alpha_H (A) - \Pi_{X_n} (\alpha_H(A))} < \norm{A} \tilde f(n),
    \end{equation}
    for $\tilde f$ depending only on $f$.
\end{proposition}
\begin{proof}
    The Lieb-Robinson bound implies $\alpha_H$ is a \textit{quasi-local map}, satisfying:
    \begin{equation}
        \norm{[\alpha_H(A),B]} < \norm{A}\norm{B} \tilde f(\text{dist(X,Y)}), \qquad A \in \mathcal A_X,\ B \in \AA_Y.
    \end{equation}
    The result then follows by Corollary 4.4. of \cite{brunoamandaI}.
\end{proof}

\subsection{Anchored TDI's}

For a TDI $H$ satisfying $\norm{H}_{X,f}<\infty$ for a finite set $X \in \mathcal P_{\text{fin}} (\ZZ)$, the sum
$$Z(s) := \sum\limits_{S \in \mathcal P_{\text{fin}}(\ZZ)} H(S)(s)$$
converges to an Hermitian element of $\AA$. An even stronger result holds:

\begin{proposition} \label{prop.localityofanchoredTDI}
    Let $H$ be such that $\norm{H}_{X,f} < \infty$, for a finite $X \in \mathcal P_{\text{fin}} (\ZZ)$. Then it holds: 
    \begin{equation}
        || Z(s) -\Pi_{X_r} Z(s) || \leq \tilde f(r) \norm{Z(s)},
    \end{equation}
    for $\tilde f$ depending only on $f$, and $\Pi$ is the map from Proposition \ref{prop.conditionalexpectation}. 
\end{proposition}
\begin{proof}
    The proof follows from Corollary 4.4 of \cite{brunoamandaI}, once we establish the commutator bounds:
        $\norm{ \left[ Z(s), A  \right] } < \norm{A} \norm{Z(s)} \tilde f(\text{dist} (X, Y)),\ A \in \AA_Y,$ for $X \cap Y = \emptyset$. This follows from direct computation (cf.\ in Lemma 5.5 of \cite{wojciech.thoulesspump}).
\end{proof}

\spacing 

\noindent As a consequence of Proposition \ref{prop.localityofanchoredTDI}, the solution $W(s)$ of
$$
W(s) = i\int_0^s du   W(u)Z(u), \qquad W(0)=\mathds 1,
$$
satisfies
\begin{equation}\label{eq.localityofunitaryatcut}
|| W(s) -\Pi_{X_r} W(s) || \leq \tilde f(r),
\end{equation}
uniformly in $|s| \le 1$, for $\tilde f$ depending only on $f$. 

\spacing


\begin{theorem} \label{thm.lgadecomposition}
    Let $\alpha_H$ be an LGA generated by TDI $H$ satisfying 
    $\norm{H}_f < \infty,\ f \in \mathcal F.$
    Consider a bipartition $L=L^{(j)},R=R^{(j)}$ of the lattice. There is a decomposition
    \begin{equation}
        \alpha_H = \alpha_{H_{L}} \otimes \alpha_{H_R} \circ \Ad{W}
    \end{equation}
    where $\alpha_{H_{L}}$ (respectively, $\alpha_{H_R}$) is an automorphism on $\AA_{L}$ ($\AA_{R}$), and $W \in \AA$. Furthermore, for any $X \in \mathcal P_{\text{fin}}(\ZZ)$, it holds
    \begin{equation}
        \norm{WAW^* - A} \le \norm {A} \tilde f(\text{dist }X,j), \qquad A \in \mathcal \AA_{X},
    \end{equation}
    for $\tilde f$ depending only on $f$. 
\end{theorem}
\begin{proof}
We can define $\alpha_{H_{R}}$ as the LGA generated by TDI
        \begin{align*}
            H_{R} (S) := \begin{cases}
                H (S),&\ S \in \mathcal P_{\text{fin}}(R), \\
                0,&\ \text{otherwise}.
            \end{cases}
        \end{align*}
Similarly for $\alpha_{H_{{L}}}$. Then the unitary $W=W(1)$ is generated by a TDI that is $\tilde f$-anchored at the boundary of the bipartition $L,R$, and the result follows from \eqref{eq.localityofunitaryatcut}. 
\end{proof}


\section{Borel cohomology} \label{appendix.borelcohomology}

	Let $G$ be a topological group equipped with Haar measure. For any left topological G-module $M$, we denote by $C^n(G,M)$ the set of measurable functions from $G \times \dots \times G \to M$, where the product runs over $n$ copies of $G$, with respect to the Borel sigma algebras in both $G\times G \times \dots \times G$ and $M$. Elements of $C^n(G,M)$ are called $n-$(measurable) cochains. The pair $(C^n(G,M),\cdot)$ is a group for every $n$, so that co-boundary group homomorphisms can be defined as
	\begin{align*}
		\nonumber    d^{n+1}:\ C^n(G,M)&\to C^{n+1}(M,G) \\
		\mu &\mapsto (d^{n+1} \mu),
	\end{align*}
	given explicitly by
	\begin{align} \label{eq.cocyclecondition}
		\nonumber (d^{n+1}\mu) (g_1,\dots,g_{n+1}) := &\mu (g_2,\dots,g_{n+1}) \\ &+ \sum\limits_{i=1}^n (-1)^i \mu(g_1,\dots,g_ig_{i+1},\dots,g_{n+1}) + (-1)^{n+1} \mu(g_1,\dots,g_n),
	\end{align}
    where we are using additive notation for the abelian module operation. The map
		\begin{equation}
			d^{n+1} \circ d^n:\ C^{n}(G,M) \to C^{n+2}(G,M)
		\end{equation}
		is identically zero, from which we conclude there is a cochain complex
	$\dots \leftarrow C^{n} \leftarrow C^{n-1} \leftarrow \dots,$
	with $n$-th (Borel) cohomology group (cf.\ \cite{moore64,stashef}) defined as 
	\begin{equation}
		H^n_{\text{Borel}}(G,M) := \dfrac{Z^n(G,M)}{B^n(G,M)},
	\end{equation}
	with 
	\begin{align*}
		Z^n(G,M) := \ker(d^{n+1}), \\
		B^n:= \begin{cases} 0,\ &\text{if } n=0, \\
			\text{im}(d^n),\ &\text{if } n\ge 1.
		\end{cases}
	\end{align*}

    Hence the n-th Borel cohomology group (Mackey-Moore cohomology group, cf.\ \cite{Cattaneo.mackeymoorecohomology}) is a group of equivalence classes of n-cocycles (elements of $Z^n(G,M)$) with respect to equivalence module n-coboundaries (elements of $B^n(G,M)$). \par
    
    \begin{remark}
        If $G$ is a finite group, Borel cohomology resumes to standard finite group cohomology, when $G$ is given the discrete topology. If G is a compact Lie group, the 2-nd Borel cohomology group $H^2_{\text{Borel}} (G,U(1))$ is a finite group \cite{moore64}.
    \end{remark}

    \begin{remark}
        SPT-phases in 1d are classified by equivalence classes of projective representations of $G$, which in turn are classified by the second Borel cohomology group. We then explicitly write the cocycle condition for $n=2$: let $c:\ G\times G \to U(1)$ be an arbitrary cochain. Then, by definition \eqref{eq.cocyclecondition}:
        \begin{equation*}
            (d^3\mu) (g,h,k) = \mu(h,k) - \mu(gh,k) + \mu(g,hk) - \mu(g,h).
        \end{equation*}
        Since 2-cocycles are elements of $\ker {d^3}$, they consequently satisfy:
        \begin{equation} \label{eq.cocyclecondition2}
            \mu(h,k) + \mu(g,hk) = \mu(gh,k) + \mu(g,h).
        \end{equation}
    \end{remark}
    

    \section{Classification of CUP-reps} \label{appendix.cups}

\begin{lemma}
    Let $(\mathbf V,\HH)$ be a continuous unitary projective representation of $G$, and assume $\mu$ and $\nu$ are Borel multipliers associated to two different lifts. Then they are related via a Borel coboundary.
\end{lemma}
\begin{proof}

    Let $v_\mu$ and $v_\nu$ be two Borel lifts. For each $g\in G$, it holds that 
    $$\mathscr P (v_\mu(g) v_\nu(g)^{-1}) = 1,$$
    since $\mathscr P$ is a group homomorphism. This implies $v_\mu(g) v_\nu(g)^{-1} = \lambda (g)$ for some $\lambda(g) \in U(1)$. Immediately,
    \begin{align*}
    \mathds 1 &= \mu(g,h) \nu(g,h)^{-1} v_\mu(g) v_\mu(h) v_\mu(gh)^{-1} v_\nu(gh) v_\nu(h)^{-1} v_\nu(g)^{-1} \\ &= \mu(g,h) \nu(g,h)^{-1} \lambda(gh)^{-1} \lambda(g)\nu(h) \cdot \mathds 1.
    \end{align*}
    Furthermore, the map $\lambda \cdot \mathds 1$ is measurable, since it is a composition of measurable maps. 
\end{proof}

Thus the map $\kappa$ from $\text{CUP}_G$ to $H^2_{\text{Borel}} (G,U(1))$ which associates a CUP-rep to its Borel cohomology class is well-defined. We proceed in proving it is an isomorphism:

    \begin{proof} [Proof of Proposition \ref{thm.isomorphism.CUP.and.H2}]
        Let $\kappa([\mathbf U]) = [\mu]$ and $\kappa([\mathbf V]) = [\nu]$. Then it holds that, for representatives $U$ and $V$, the Borel representation $U \otimes V$ has multiplier $\mu + \nu$. Hence $\kappa([\mathbf U \otimes \mathbf V]) = [\mu + \nu] = [\mu] + [\nu]$, and $f$ is a group homomorphism. \par Furthermore, each cohomology class can be realized as a finite dimensional continuous projective representation of $G$ (proof of Lemma \ref{lem.surjective}), hence $\kappa$ is surjective. It is also injective since, if $\kappa([\mathbf U]) = [1]$, then one can choose phase $\beta: G \to U(1)$ and unitaries $W_g := \beta(g)\mathds 1$ such that $W_g \otimes U$ is a lift of a linear representation of $G$. \par
    \end{proof}


    \par



\section*{Data Availability}
Data sharing is not applicable to this article as no new data were created or analyzed in this study.
 
	\bibliographystyle{myplainurl.bst}
	\bibliography{main}
	
\end{document}

%% file: myPackage.tex
\usepackage{graphicx,xcolor,textpos}

\graphicspath{{./img/}{./Images/}}

\usepackage{}
\usepackage{helvet}
\usepackage[english]{babel}
\usepackage[normalem]{ulem}
\usepackage{amsmath}
\usepackage{amsthm}
\usepackage{thmtools,thm-restate}
\usepackage{bbm}
\usepackage{amssymb}
\usepackage{hyperref}
\usepackage{relsize}
\usepackage[margin=0.7in]{geometry}
\usepackage{physics}
\usepackage[shortlabels]{enumitem}
\usepackage{mathtools}
\usepackage{changepage}
\usepackage{caption}
\usepackage{subcaption}
\usepackage{verbatim}
\usepackage{url}
\usepackage{standalone}
\usepackage{tikz}
\usetikzlibrary{calc,patterns,angles,quotes}
\usepackage{dsfont}
\usepackage{mathrsfs}
\usepackage[affil-it]{authblk}
\usepackage{titlesec}

\makeatletter
\newcommand\Label[1]{&\refstepcounter{equation}(\theequation)\ltx@label{#1}&}
\makeatother

\newcommand{\stkout}[1]{\ifmmode\text{\sout{\ensuremath{#1}}}\else\sout{#1}\fi}

\let\originalleft\left
\let\originalright\right
\renewcommand{\left}{\mathopen{}\mathclose\bgroup\originalleft}
\renewcommand{\right}{\aftergroup\egroup\originalright}

\newcommand{\UU}{\mathcal U}

\newcommand{\HH}{\mathcal H}
\newcommand{\ZZ}{\mathbb Z}
\newcommand{\CC}{\mathbb C}

\renewcommand{\AA}{\mathcal A}

\newcommand{\NN}{\mathbb{N}}

\newcommand{\Ad}[1]{\textrm{Ad}\left(#1\right)}

\newtheoremstyle{exampstyle}
{15pt} 
{15pt} 
{\itshape} 
{} 
{\bfseries} 
{.} 
{.5em} 
{} 

\theoremstyle{exampstyle}

\newtheorem{theorem}{Theorem}[section]
\newtheorem{definition}[theorem]{Definition}
\newtheorem{lemma}[theorem]{Lemma}
\newtheorem{remark}[theorem]{Remark}

\newtheorem{corollary}[theorem]{Corollary}

\newtheorem{proposition}[theorem]{Proposition}